\documentclass[a4paper]{article}
\usepackage{lmodern}
\pdfoutput=1
\usepackage{graphicx,wrapfig}
\usepackage{lipsum}
\usepackage{amsmath,amssymb,mathrsfs}
\usepackage{hyperref}
\usepackage{url}
\usepackage{mathtools}
\usepackage{changepage}
\usepackage[ruled,vlined]{algorithm2e}
\usepackage{amsmath}
\usepackage{multirow}
\usepackage{wrapfig}
\usepackage{subfig}
\usepackage{color}
\usepackage{epsfig,enumerate,amsmath,amsfonts,amssymb,amsthm,mathrsfs,ifpdf}
\usepackage{indentfirst,relsize}
\usepackage{setspace,graphicx}
\usepackage{latexsym}
\usepackage[all]{xy}
\usepackage[dvipsnames]{pstricks}
\usepackage{pst-grad} 
\usepackage{pst-plot} 
\usepackage[margin = 2.50cm]{geometry}

\usepackage{authblk} 
\usepackage{algpseudocode}
\usepackage{color,soul}
\usepackage{enumitem}

\usepackage{orcidlink}
\usepackage{tabularx}
\usepackage{pst-grad} 
\usepackage{pst-plot} 
\usepackage{enumitem}

\usepackage{orcidlink}

\usepackage{tikz}
\usetikzlibrary{arrows.meta,positioning}

\newcommand{\remove}[1]{}

\newtheorem{theorem}{Theorem}
\newtheorem{lemma}[theorem]{Lemma}

\newtheorem{remark}{Remark}

\usepackage{todonotes}
\usepackage[most]{tcolorbox}

\title{Roman Domination on Circular-Convex, Triad-Convex Bipartite Graphs and \(P_4\)-Tidy Graphs}

\author[1]{Gautam K. Das\footnote{gkd@iitg.ac.in}}
\author[1]{Kamal Santra \orcidlink{0009-0006-5997-1452} \footnote{kamal.7.2013@gmail.com, kamal.santra@iitg.ac.in}}
\affil[1]{Department of Mathematics\\
	
	Indian Institute of Technology Guwahati\\
	
	Guwahati, 781039, Assam, India}

\date{}

\begin{document}

	\maketitle
\begin{abstract}
	The Roman Domination Problem (RDP) on a graph \(G=(V,E)\) asks for a labeling function \(f:V\rightarrow\{0,1,2\}\) such that every vertex assigned value \(0\) is adjacent to a vertex assigned value \(2\). The objective is to minimize the total weight \(\sum_{v\in V} f(v)\); this minimum value is the Roman domination number of \(G\), denoted by \(\gamma_R(G)\). In this paper, we study RDP on graph classes motivated by convexity and induced-\(P_4\) structure. First, we consider circular-convex bipartite graphs, a natural superclass of convex bipartite graphs, where RDP is already known to be polynomial-time solvable. Assuming that a circular-convex representation is given, we compute \(\gamma_R(G)\) in \(O(n^6)\) time by cutting the circular order, separating interval and wrap-around vertices, and branching over at most two wrap-around vertices assigned value \(2\). Second, we study triad-convex bipartite graphs, a restricted subclass of tree-convex bipartite graphs whose convexity tree is a subdivision of \(K_{1,3}\). Although RDP is hard on broader tree-convex subclasses such as star-convex and comb-convex bipartite graphs, we show that \(\gamma_R(G)\) can be computed in \(O(n^7)\) time on triad-convex bipartite graphs. Finally, we study \(P_4\)-tidy graphs, which properly extend cographs. Using the Giakoumakis et al. structural decomposition of \(P_4\)-tidy graphs, we give a direct, exact algorithm that computes \(\gamma_R(G)\) in \(O(n+m)\) time. These results extend the algorithmic boundary of Roman domination on convexity-based bipartite graphs and \(P_4\)-structured graph classes.
\end{abstract}



\section{Introduction}\label{sec:introduction}

Domination is a central topic in graph theory and graph algorithms, with applications in facility location, communication networks, wireless sensor networks, and social network analysis. A set \(D\subseteq V(G)\) is a dominating set of a graph \(G\) if every vertex outside \(D\) has a neighbour in \(D\). The minimum size of such a set is the domination number \(\gamma(G)\). The classical theory of domination and its variants is treated in detail in the books and surveys of Haynes et al.~\cite{haynes1998fundamentals}, Hedetniemi and Laskar~\cite{hedetniemi1991bibliography}, and Haynes et al.~\cite{haynes2020topics}.

Roman domination is one of the most studied labelled variants of domination. It was introduced by Cockayne et al.~\cite{cockayne2004roman}, and was motivated by the classical problem of defending the Roman Empire with limited military resources~\cite{stewart1999defend}. A Roman dominating function of a graph \(G\) is a function \(f:V(G)\rightarrow\{0,1,2\}\) such that every vertex \(v\) with \(f(v)=0\) has a neighbour \(u\) with \(f(u)=2\). The weight of \(f\) is \(w(f)=\sum_{v\in V(G)}f(v)\), and the Roman domination number \(\gamma_R(G)\) is the minimum possible weight of a Roman dominating function of \(G\). The corresponding decision problem asks whether \(\gamma_R(G)\leq k\), where \(k\) is part of the input.

The Roman Domination Problem is computationally hard in general. It is NP-complete on general graphs~\cite{dreyer2000applications}, and it remains NP-complete on bipartite graphs, split graphs, and planar graphs~\cite{cockayne2004roman,mcrae2002private}. Poureidi and Fathali~\cite{poureidi2023algorithmic} showed that the problem is also NP-complete on circle graphs. On the positive side, Liedloff et al.~\cite{liedloff2008efficient} gave efficient algorithms for Roman domination on several graph classes, including interval graphs, cographs, \(D\)-octopus graphs, AT-free graphs, and distance-hereditary graphs. Padamutham and Palagiri~\cite{padamutham2020algorithmic} studied further algorithmic aspects of Roman domination; in particular, they showed hardness for star-convex and comb-convex bipartite graphs and gave polynomial-time algorithms for bounded-treewidth graphs, chain graphs, and threshold graphs. Roman domination has also been studied on regular graphs and other special graph families; see, for example, Fu et al.~\cite{fu2009romanregular}.

Figure~\ref{fig:rdp-status} summarizes the status of Roman domination on
several graph classes related to this paper. The arrows indicate subclass
relations. Red boxes indicate known NP-complete classes, blue dashed boxes
indicate classes where polynomial-time or linear-time algorithms were already
known, and yellow boxes highlight the classes treated in this paper.

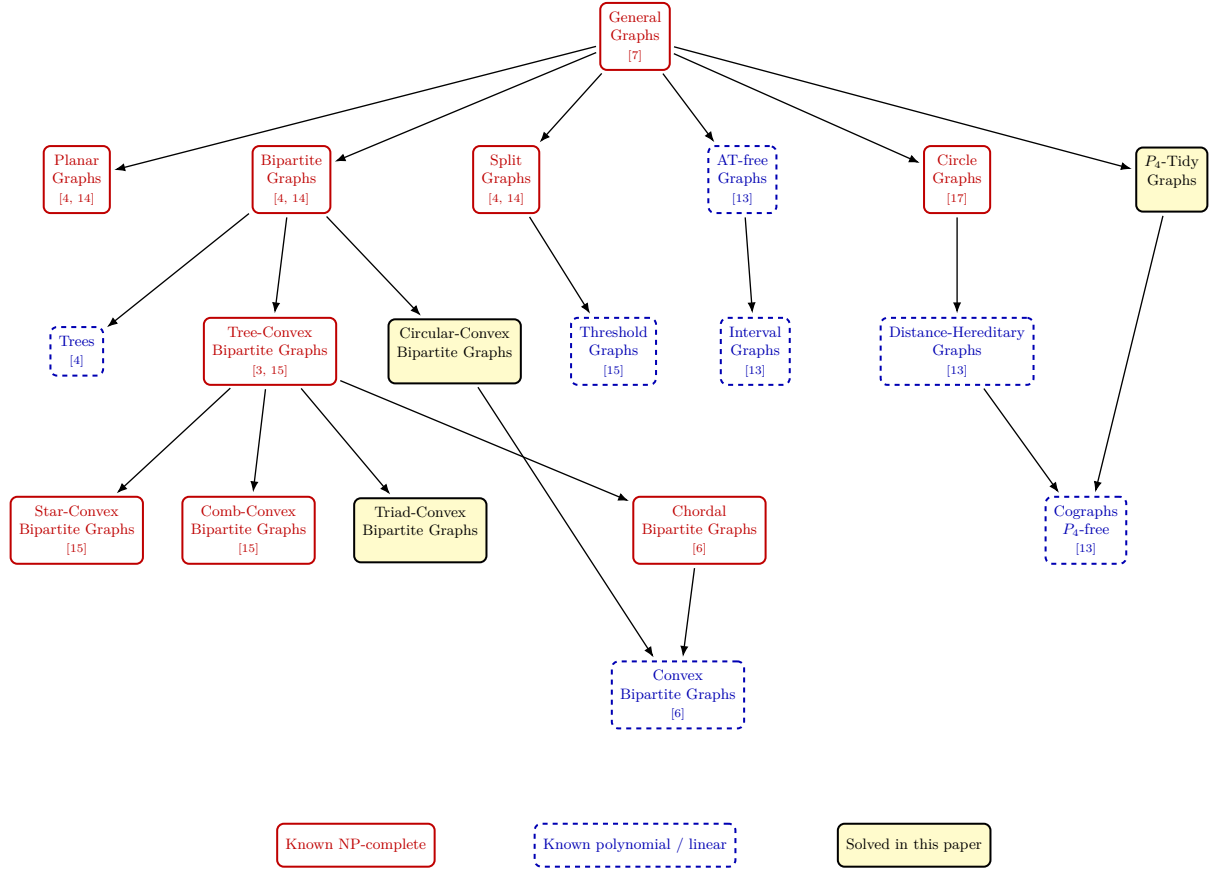
\begin{figure}[h]
	\centering
	\resizebox{\textwidth}{!}{
		\begin{tikzpicture}[
			x=0.9cm,
			y=1.5cm,
			>=Latex,
			every node/.style={font=\small},
			edge/.style={->, thick, shorten >=2pt, shorten <=2pt},
			hard/.style={
				draw=red!75!black,
				text=red!75!black,
				very thick,
				rounded corners,
				align=center,
				inner sep=5pt,
				minimum height=9mm
			},
			known/.style={
				draw=blue!70!black,
				text=blue!70!black,
				dashed,
				very thick,
				rounded corners,
				align=center,
				inner sep=5pt,
				minimum height=9mm
			},
			paper/.style={
				draw=black,
				fill=yellow!25,
				very thick,
				rounded corners,
				align=center,
				inner sep=5pt,
				minimum height=9mm
			}
			]
			
			\node[hard] (general) at (0,0) {General\\Graphs\\{\scriptsize\cite{dreyer2000applications}}};
			
			\node[hard]  (planar) at (-13.0,-2.0) {Planar\\Graphs\\{\scriptsize\cite{cockayne2004roman,mcrae2002private}}};
			\node[hard]  (bip)    at (-8.0,-2.0)  {Bipartite\\Graphs\\{\scriptsize\cite{cockayne2004roman,mcrae2002private}}};
			\node[hard]  (split)  at (-3.0,-2.0)  {Split\\Graphs\\{\scriptsize\cite{cockayne2004roman,mcrae2002private}}};
			\node[known] (atfree) at (2.5,-2.0)   {AT-free\\Graphs\\{\scriptsize\cite{liedloff2008efficient}}};
			\node[hard]  (circle) at (7.5,-2.0)   {Circle\\Graphs\\{\scriptsize\cite{poureidi2023algorithmic}}};
			\node[paper] (p4tidy) at (12.5,-2.0)  {\(P_4\)-Tidy\\Graphs\\};
			
			\node[known] (trees)     at (-13.0,-4.4) {Trees\\{\scriptsize\cite{cockayne2004roman}}};
			\node[hard]  (treeconv)  at (-8.5,-4.4)  {Tree-Convex\\Bipartite Graphs\\{\scriptsize\cite{bao2012treeconvex,padamutham2020algorithmic}}};
			\node[paper] (circular)  at (-4.2,-4.4)  {Circular-Convex\\Bipartite Graphs\\};
			\node[known] (threshold) at (-0.5,-4.4)  {Threshold\\Graphs\\{\scriptsize\cite{padamutham2020algorithmic}}};
			\node[known] (interval)  at (2.8,-4.4)   {Interval\\Graphs\\{\scriptsize\cite{liedloff2008efficient}}};
			\node[known] (dh)        at (7.5,-4.4)   {Distance-Hereditary\\Graphs\\{\scriptsize\cite{liedloff2008efficient}}};
			
			\node[hard]  (star)    at (-13.0,-6.9) {Star-Convex\\Bipartite Graphs\\{\scriptsize\cite{padamutham2020algorithmic}}};
			\node[hard]  (comb)    at (-9.0,-6.9)  {Comb-Convex\\Bipartite Graphs\\{\scriptsize\cite{padamutham2020algorithmic}}};
			\node[paper] (triad)   at (-5.0,-6.9)  {Triad-Convex\\Bipartite Graphs\\};
			\node[hard]  (chordal) at (1.5,-6.9)  {Chordal\\Bipartite Graphs\\{\scriptsize\cite{das2026roman}}};
			\node[known] (cograph) at (10.5,-6.9)  {Cographs\\\(P_4\)-free\\{\scriptsize\cite{liedloff2008efficient}}};
			
			\node[known] (convex)  at (1.0,-9.2)  {Convex\\Bipartite Graphs\\{\scriptsize\cite{das2026roman}}};
			
			\draw[edge] (general) -- (planar);
			\draw[edge] (general) -- (bip);
			\draw[edge] (general) -- (split);
			\draw[edge] (general) -- (atfree);
			\draw[edge] (general) -- (circle);
			\draw[edge] (general) -- (p4tidy);
			
			\draw[edge] (split) -- (threshold);
			\draw[edge] (atfree) -- (interval);
			\draw[edge] (circle) -- (dh);
			\draw[edge] (dh) -- (cograph);
			\draw[edge] (p4tidy) -- (cograph);
			
			\draw[edge] (bip) -- (trees);
			\draw[edge] (bip) -- (treeconv);
			\draw[edge] (bip) -- (circular);
			
			\draw[edge] (treeconv) -- (star);
			\draw[edge] (treeconv) -- (comb);
			\draw[edge] (treeconv) -- (triad);
			\draw[edge] (treeconv) -- (chordal);
			
			\draw[edge] (circular) -- (convex);
			\draw[edge] (chordal) -- (convex);
			
			\node[hard]  (leg1) at (-6.5,-11.3) {Known NP-complete};
			\node[known] (leg2) at (0,-11.3)    {Known polynomial / linear};
			\node[paper] (leg3) at (6.5,-11.3)  {Solved in this paper};
			
		\end{tikzpicture}
	}
	\caption{Status of the Roman Domination Problem on selected graph classes. The classes treated in this paper are highlighted.}
	\label{fig:rdp-status}
\end{figure}

Convexity-based bipartite graph classes form a natural setting for domination problems. A bipartite graph \(G=(X\cup Y,E)\) is convex with respect to \(X\) if \(X\) admits a linear ordering such that the neighborhood of every vertex in \(Y\) is an interval in this ordering. More generally, \(G\) is tree-convex if there is a tree \(T\) on \(X\) such that the neighborhood of every vertex in \(Y\) induces a connected subtree of \(T\)~\cite{bao2012treeconvex}. Important subclasses arise by restricting the shape of \(T\). If \(T\) is a path, a star, or a comb, then one obtains convex, star-convex, or comb-convex bipartite graphs, respectively. Chordal bipartite graphs also fit into the tree-convex framework. Domination on convex, circular-convex, and triad-convex bipartite graphs has been studied by Bang-Jensen et al.~\cite{bangjensen1999domination} and Pandey and Panda~\cite{pandey2019domination}.

Das et al.~\cite{das2026roman} studied Roman domination on convex bipartite graphs and gave a dynamic programming algorithm for computing \(\gamma_R(G)\) in \(O(n^3)\) time. They also settled the complexity of Roman domination on chordal bipartite graphs. Their algorithm for convex bipartite graphs shows that the interval ordering can be exploited to keep only compact boundary information. However, the complexity of Roman domination was not settled for some natural classes beyond the linear-order setting, including circular-convex bipartite graphs and triad-convex bipartite graphs. This motivates the question of whether the dynamic programming approach for convex bipartite graphs can be extended to broader convexity-based graph classes.

Another important direction is given by graph classes defined through induced \(P_4\)-structure. Cographs are precisely the \(P_4\)-free graphs, and Roman domination is already known to be solvable in \(O(n+m)\) time on cographs~\cite{liedloff2008efficient}. The class of \(P_4\)-tidy graphs properly extends cographs and admits a decomposition due to Giakoumakis et al.~\cite{giakoumakis1997p_4}. This raises a natural question: can the linear-time solvability known for cographs be extended to the larger \(P_4\)-tidy class by using its decomposition structure?

\subsection{Our Contribution}

In this paper, we continue the study of Roman domination on structured graph classes. Our results move in three algorithmic directions: two for convexity-based bipartite graph classes and one for a graph class defined by induced-\(P_4\) structure.

First, we consider circular-convex bipartite graphs. In these graphs, one part of the bipartition has a circular ordering such that the neighborhood of every vertex in the other part appears consecutively on the circle. This class naturally generalizes convex bipartite graphs, for which Roman domination is already known to be polynomial-time solvable~\cite{das2026roman}. We show that this tractability also extends to the circular setting. Assuming that a circular-convex representation is given, we present an \(O(n^6)\)-time dynamic-programming algorithm for computing \(\gamma_R(G)\). The main idea is to cut the circular ordering into a linear ordering, distinguish interval vertices from wrap-around vertices, and branch over at most two wrap-around vertices that are assigned value \(2\). Once this choice is fixed, the remaining instance can be solved by a boundary-aware convex dynamic program.

Second, we consider triad-convex bipartite graphs. These graphs form a restricted subclass of tree-convex bipartite graphs, where the underlying convexity tree is a subdivision of \(K_{1,3}\). Since Roman domination is NP-complete on broader tree-convex subclasses, such as star-convex and comb-convex bipartite graphs~\cite{padamutham2020algorithmic}, it is natural to ask whether the additional three-arm structure of triad-convex bipartite graphs leads to a polynomial-time algorithm. We answer this question affirmatively. Assuming that a triad-convex representation is given, we design an \(O(n^7)\)-time algorithm for computing \(\gamma_R(G)\). The algorithm branches over at most three central \(Y\)-vertices assigned value \(2\), fixes the value of the center of the triad, solves the three arms independently using boundary-aware convex dynamic programming, and then combines the resulting arm states.

Finally, we study \(P_4\)-tidy graphs. This class properly contains cographs, for which Roman domination is already known to be linear-time solvable~\cite{liedloff2008efficient}. We give a direct exact algorithm for \(P_4\)-tidy graphs using the Giakoumakis et al. decomposition~\cite{giakoumakis1997p_4}. The algorithm stores two values at each node of the decomposition tree, namely \(\gamma_R(G)\) and the minimum Roman domination value among assignments that use at least one vertex of value \(2\). It handles disconnected nodes, complement disconnected nodes, spider and quasi-spider nodes, and the exceptional graphs \(K_1\), \(P_5\), \(\overline{P_5}\), and \(C_5\). This gives an \(O(n+m)\)-time algorithm for computing \(\gamma_R(G)\) on \(P_4\)-tidy graphs.

The rest of the paper is organized as follows. Section~\ref{sec:preliminaries} gives the notation and definitions used throughout the paper. Section~\ref{sec:rd-circular} presents the algorithm for circular-convex bipartite graphs. Section~\ref{sec:rd-triad} gives the algorithm for triad-convex bipartite graphs. Section~\ref{sec:rd-p4-tidy} gives the linear-time lgorithm for \(P_4\)-tidy graphs. Finally, Section~\ref{sec:conclusion} concludes the paper.

\section{Preliminaries}\label{sec:preliminaries}

All graphs considered in this paper are finite, simple, and undirected. For a graph \(G\), we denote its vertex set and edge set by \(V(G)\) and \(E(G)\), respectively. For a vertex \(v\in V(G)\), the open neighborhood of \(v\) is denoted by \(N_G(v)\), or simply by \(N(v)\) when the graph is clear from the context. A vertex of degree one is called a pendant vertex. For a positive integer \(n\), we write \([n]=\{1,2,\ldots,n\}\).

A Roman dominating function, or RDF, of a graph \(G\) is a function \(f:V(G)\rightarrow\{0,1,2\}\) such that every vertex \(v\) with \(f(v)=0\) has a neighbour \(u\in N(v)\) with \(f(u)=2\). The weight of \(f\) is \(w(f)=\sum_{v\in V(G)}f(v)\). The Roman domination number of \(G\), denoted by \(\gamma_R(G)\), is the minimum weight of an RDF of \(G\). The decision version of Roman Domination asks, given a graph \(G\) and an integer \(k\), whether \(\gamma_R(G)\leq k\)~\cite{cockayne2004roman}.

Let \(G=(X\cup Y,E)\) be a bipartite graph. We say that \(G\) is convex with respect to \(X\) if there is a linear ordering \(X=\{x_1,x_2,\ldots,x_m\}\) such that, for every \(y\in Y\), the neighborhood \(N(y)\) is an interval in this ordering. Thus, for every \(y\in Y\), there are integers \(l(y)\) and \(r(y)\), with \(1\leq l(y)\leq r(y)\leq m\), such that \(N(y)=\{x_{l(y)},x_{l(y)+1},\ldots,x_{r(y)}\}\). We call \(l(y)\) and \(r(y)\) the left and right endpoints of \(y\), respectively. Convex bipartite graphs are also called line-convex bipartite graphs in some papers.

A bipartite graph \(G=(X\cup Y,E)\) is circular-convex with respect to \(X\) if the vertices of \(X\) admit a circular ordering such that \(N(y)\) is consecutive on the circle for every \(y\in Y\)~\cite{bangjensen1999domination,pandey2019domination}. In the algorithm for circular-convex bipartite graphs, we assume that such a representation is given. After cutting the circle between \(x_m\) and \(x_1\), each vertex \(y\in Y\) is represented in one of two forms. It is an interval vertex if \(N(y)=\{x_{l(y)},x_{l(y)+1},\ldots,x_{r(y)}\}\) with \(l(y)\leq r(y)\). It is a wrap-around vertex if \(N(y)=\{x_{l(y)},x_{l(y)+1},\ldots,x_m\}\cup\{x_1,x_2,\ldots,x_{r(y)}\}\) with \(l(y)>r(y)\). We denote the set of interval vertices by \(Y^I\), and the set of wrap-around vertices by \(Y^W\).

More generally, a bipartite graph \(G=(X\cup Y,E)\) is tree-convex with respect to \(X\) if there is a tree \(T\) on \(X\) such that \(N(y)\) induces a connected subtree of \(T\) for every \(y\in Y\)~\cite{bao2012treeconvex}. Several important subclasses are obtained by restricting the shape of \(T\). If \(T\) is a path, a star, or a comb, then one obtains convex, star-convex, or comb-convex bipartite graphs, respectively.

A bipartite graph \(G=(X\cup Y,E)\) is triad-convex with respect to \(X\) if there is a tree \(T\) on \(X\), where \(T\) is a subdivision of \(K_{1,3}\), such that \(N(y)\) induces a connected subtree of \(T\) for every \(y\in Y\)~\cite{pandey2019domination}. We call \(T\) a triad representation of \(G\). Let \(c\) be the unique branching vertex of \(T\). Removing \(c\) leaves three paths, called arms, denoted by \(P_1,P_2,P_3\). We write \(P_k=\{x_{k,1},x_{k,2},\ldots,x_{k,m_k}\}\), ordered from the center outward. A vertex \(y\in Y\) is an arm vertex on \(P_k\) if \(N(y)\subseteq V(P_k)\); in this case, \(N(y)\) is an interval on the path \(P_k\). A vertex \(y\in Y\) is a central \(Y\)-vertex if \(c\in N(y)\). In this case, since \(N(y)\) induces a connected subtree of the triad, we may write \(N(y)=\{c\}\cup \bigcup_{k=1}^{3}\{x_{k,1},x_{k,2},\ldots,x_{k,r_k(y)}\}\), where \(0\leq r_k(y)\leq m_k\). We denote the set of central \(Y\)-vertices by \(Y^C\), and the set of arm vertices on \(P_k\) by \(Y_k\).

An induced \(P_4\) is a chordless path on four vertices. Let \(A\) be an
induced \(P_4\) of a graph \(G\). A vertex
\(v\in V(G)\setminus V(A)\) is called a \emph{partner} of \(A\) if the
subgraph \(G[V(A)\cup\{v\}]\) contains at least two induced \(P_4\)'s. A
graph \(G\) is called \(P_4\)-tidy if every induced \(P_4\) of \(G\) has at
most one partner.

A graph \(G\) is called a \emph{spider} if \(V(G)\) can be partitioned into
three sets \((S,K,R)\), where \(S\) is an independent set, \(K\) is a clique,
\(|S|=|K|\), every vertex of \(R\) is adjacent to all vertices of \(K\) and to
no vertex of \(S\), and there is a bijection between \(S\) and \(K\). In a
\emph{thin spider}, each vertex of \(S\) is adjacent only to its corresponding
vertex in \(K\), whereas in a \emph{thick spider}, each vertex of \(S\) is
adjacent to every vertex of \(K\) except its corresponding vertex. The set
\(R\) is called the \emph{head} of the spider. A \emph{quasi-spider} is obtained
from a spider by replacing exactly one vertex of \(S\) or \(K\) with either
\(K_2\) or \(\overline{K_2}\), while preserving its adjacencies to all remaining
vertices.

The following characterization of \(P_4\)-tidy graphs, due to Giakoumakis
et al.~\cite{giakoumakis1997p_4}, will be used throughout the paper. A graph \(G\) is \(P_4\)-tidy if and
only if one of the following holds:
\begin{itemize}
	\item \(G\) is isomorphic to one of \(K_1\), \(P_5\), \(\overline{P_5}\), or \(C_5\);
	\item \(G\) is disconnected;
	
	\item \(\overline{G}\) is disconnected;
	\item \(G\) is a spider or a quasi-spider.
\end{itemize}

In the spider and quasi-spider cases, the subgraph induced by the head \(R\) is
itself \(P_4\)-tidy. Consequently, every \(P_4\)-tidy graph admits a recursive
decomposition, which we use in Section~\ref{sec:rd-p4-tidy}.


\section{Roman Domination on Circular-Convex Bipartite Graphs}
\label{sec:rd-circular}

In this section, we extend the dynamic programming approach for Roman
Domination on convex bipartite graphs due to Das et
al.~\cite{das2026roman} to circular-convex bipartite graphs. Let
\(G=(X\cup Y,E)\) be a bipartite graph, where
\(X=\{x_1,x_2,\ldots,x_m\}\) is given in a circular order. We assume that a
circular-convex representation of \(G\) is given. Thus, for every vertex
\(y\in Y\), the neighbourhood \(N(y)\) appears consecutively on the circle.

We cut the circle between \(x_m\) and \(x_1\), and view the circular order as
the linear order \(x_1,x_2,\ldots,x_m\). After this cut, each vertex of \(Y\)
is one of two types. If \(N(y)\) does not cross the cut, then
\(N(y)=\{x_{l(y)},x_{l(y)+1},\ldots,x_{r(y)}\}\), where
\(l(y)\leq r(y)\). We call such a vertex an \emph{interval vertex}. If
\(N(y)\) crosses the cut, then
\(N(y)=\{x_{l(y)},\ldots,x_m\}\cup\{x_1,\ldots,x_{r(y)}\}\), where
\(l(y)>r(y)\). We call such a vertex a \emph{wrap-around vertex}. Let
\(Y^I\) denote the set of interval vertices, and let \(Y^W\) denote the set
of wrap-around vertices. A vertex \(y\) with \(N(y)=X\) may be treated as an
interval vertex with \(l(y)=1\) and \(r(y)=m\).

The vertices in \(Y^I\) form an ordinary convex bipartite instance with
respect to the linear order \(x_1,x_2,\ldots,x_m\). Therefore, the
convex-bipartite dynamic programming algorithm of Das et
al.~\cite{das2026roman} can be used as the base routine after the wrap-around
vertices are handled. The difficulty comes from the vertices in \(Y^W\),
because their neighbourhoods are not intervals in this linear order. We
handle them by branching over the wrap-around vertices that receive Roman
value \(2\).

\subsection{A Structural Observation}

We first show that, in some optimal solution, only a small number of
wrap-around vertices need to receive value \(2\).

\begin{lemma}
	\label{lem:circular-two-wrap}
	There exists an optimal Roman dominating function in which at most two
	vertices of \(Y^W\) are assigned value \(2\).
\end{lemma}

\begin{proof}
	Let \(f\) be an optimal Roman dominating function, and let
	\(S_f=\{y\in Y^W:f(y)=2\}\). If \(|S_f|\leq 2\), then there is nothing to
	prove. Suppose that \(|S_f|>2\).
	
	Every vertex \(y\in S_f\) has \(N(y)=\{x_{l(y)},\ldots,x_m\}\cup\{x_1,\ldots,x_{r(y)}\}\). Among the
	vertices of \(S_f\), choose a vertex \(y_a\) with the minimum left endpoint
	\(l(y_a)\), and choose a vertex \(y_b\) with the maximum right endpoint
	\(r(y_b)\). Then \(N(y_a)\cup N(y_b)\) contains every vertex of \(X\) that
	is dominated by some vertex of \(S_f\). Indeed, among all suffixes
	\(\{x_{l(y)},\ldots,x_m\}\), the suffix corresponding to \(y_a\) is the
	largest because \(l(y_a)\) is minimum. Similarly, among all prefixes
	\(\{x_1,\ldots,x_{r(y)}\}\), the prefix corresponding to \(y_b\) is the
	largest because \(r(y_b)\) is maximum. Hence, \(y_a\) and \(y_b\) together
	dominate every vertex of \(X\) that is dominated by at least one vertex of
	\(S_f\).
	
	Now, change the value of every vertex in
	\(S_f\setminus\{y_a,y_b\}\) from \(2\) to \(1\). These vertices remain safe,
	since a vertex assigned value \(1\) does not need to be dominated. No vertex
	of \(X\) loses domination, because \(y_a\) and \(y_b\) together dominate all
	vertices of \(X\) that were dominated by the vertices of \(S_f\). Moreover,
	vertices of \(Y\) do not dominate other vertices of \(Y\), since \(G\) is
	bipartite with bipartition \((X,Y)\). Thus, the modified function is still a
	Roman dominating function, but it has a smaller weight, contradicting the
	optimality of \(f\). Therefore, some optimal Roman dominating functions use
	at most two wrap-around vertices with value \(2\).
\end{proof}

\subsection{Dynamic Programming for a Fixed Wrap-Around Choice}

By Lemma~\ref{lem:circular-two-wrap}, there exists an optimal Roman
dominating function in which at most two vertices of \(Y^W\) receive value
\(2\). We therefore consider every subset \(S\subseteq Y^W\) with
\(|S|\leq 2\). In the branch corresponding to \(S\), every vertex of \(S\)
is forced to receive value \(2\), while every vertex of \(Y^W\setminus S\)
is forbidden from receiving value \(2\). Thus, \(S\) represents the exact
set of wrap-around vertices assigned value \(2\) in the current branch. The
vertices of \(S\) contribute \(2|S|\) to the weight of the solution. This
contribution is not included in the dynamic programming table and is added
after the dynamic program terminates.

\paragraph*{\underline{\textbf{Pre-dominated vertices of \texorpdfstring{\(X\)}}}}

The vertices of \(S\) may already dominate some vertices of \(X\). For every
\(x_i\in X\), define \(\operatorname{pre}_S(x_i)=0\) if \(x_i\) is adjacent
to at least one vertex of \(S\), and define
\(\operatorname{pre}_S(x_i)=1\) otherwise. Therefore,
\(\operatorname{pre}_S(x_i)=0\) exactly when
\(x_i\in\bigcup_{y\in S}N(y)\). In this case, if \(x_i\) is assigned value
\(0\), its Roman domination requirement is already satisfied by a forced
value-\(2\) vertex of \(S\). On the other hand,
\(\operatorname{pre}_S(x_i)=1\) means that \(x_i\) is not pre-dominated by
\(S\). Hence, if such a vertex is assigned value \(0\), it must eventually
be dominated by an interval vertex of \(Y^I\) assigned value \(2\).

The convention that \(\operatorname{pre}_S(x_i)=0\) denotes a pre-dominated
vertex is useful in the transitions: the value \(0\) indicates that no
further domination obligation needs to be created for \(x_i\).

\paragraph*{\underline{\textbf{Processing order.}}}

For the fixed set \(S\), we run a modified version of the convex-bipartite
dynamic program of Das et al.~\cite{das2026roman} on the subgraph induced by
\(X\cup Y^I\). Recall that every vertex \(y\in Y^I\) has an interval
neighbourhood \(N(y)=\{x_{l(y)},x_{l(y)+1},\ldots,x_{r(y)}\}\). For every
\(i\in\{1,\ldots,m\}\), let
\(Y_i^I=\{y\in Y^I:r(y)=i\}\). The vertices are processed in the order
\(x_1,Y_1^I,x_2,Y_2^I,\ldots,x_m,Y_m^I\). Thus, immediately after processing \(x_i\), we process all interval vertices
whose right endpoint is \(i\). The vertices inside \(Y_i^I\) may be processed
in any fixed order.

This processing order ensures that, when a vertex \(y\in Y_i^I\) is
processed, every vertex of \(N(y)\) has already been processed. Indeed,
\(N(y)=\{x_{l(y)},\ldots,x_i\}\). Consequently, whether \(y\) can receive
value \(0\) can be decided from the value-\(2\) vertices of \(X\) already
recorded in the state. Moreover, if \(y\) receives value \(2\), all processed
vertices of \(X\) that it dominates are known at that moment.

\paragraph*{\underline{\textbf{Pending vertices}}}

Suppose that a processed vertex \(x_j\) has been assigned value \(0\). If
\(\operatorname{pre}_S(x_j)=0\), then \(x_j\) is already dominated by \(S\).
Otherwise, \(x_j\) must be dominated by a processed or future interval vertex
of \(Y^I\) assigned value \(2\). Until such a vertex is selected, we call
\(x_j\) \emph{pending}.

It is enough to remember only the earliest pending vertex. To see this,
suppose that \(x_p\) is the earliest pending vertex. Every other pending
vertex has an index at least \(p\). A future interval vertex \(y\), when
processed, has a neighbourhood of the form
\(N(y)=\{x_{l(y)},\ldots,x_{r(y)}\}\), where all currently processed vertices
have index at most \(r(y)\). If \(l(y)\leq p\), then \(y\) contains not only
\(x_p\), but also every later pending vertex. Hence, assigning value \(2\) to
\(y\) removes all current pending obligations. If \(l(y)>p\), then \(y\)
does not dominate the earliest pending vertex, so the branch must continue
to remember \(p\), even if \(y\) dominates some later pending vertices. Thus,
the earliest pending index contains all the information required by future
transitions.

\paragraph*{\underline{\textbf{State definition}}}

As in the convex-bipartite algorithm, we distinguish a special state value
from the numerical value \(+\infty\) used for an infeasible table entry. Let
\(a_\infty=m+1\) and \(p_\infty=m+1\). The symbol \(a_\infty\) means that no
processed vertex of \(X\) has been assigned value \(2\), while
\(p_\infty\) means that no processed vertex of \(X\) is pending. Although
both special values are represented numerically by \(m+1\), their different
symbols make their roles in the state explicit. The value \(+\infty\) is
reserved exclusively for an infeasible dynamic programming entry.

A state is a triple \((a,b,p)\). The coordinates have the following
interpretations.

\begin{itemize}
	\item The value \(a\) is the smallest index \(t\) such that the processed
	vertex \(x_t\) has been assigned value \(2\). If no processed vertex of
	\(X\) has value \(2\), then \(a=a_\infty\).
	
	\item The value \(b\) is the largest index \(t\) such that the processed
	vertex \(x_t\) has been assigned value \(2\). If no processed vertex of
	\(X\) has value \(2\), then \(b=0\).
	
	\item The value \(p\) is the smallest index of a processed vertex \(x_p\)
	such that \(f(x_p)=0\), \(\operatorname{pre}_S(x_p)=1\), and \(x_p\) has
	not yet been dominated by a processed vertex of \(Y^I\) assigned value
	\(2\). If no such vertex exists, then \(p=p_\infty\).
\end{itemize}

The table entry \(D(a,b,p)\) stores the minimum weight of a partial assignment
on the processed vertices of \(X\cup Y^I\) that realizes the state
\((a,b,p)\). The contribution \(2|S|\) of the forced vertices in \(S\) is
not included in \(D(a,b,p)\). An entry has value \(+\infty\) when no feasible
partial assignment realizes the corresponding state.

At every processing stage, the table satisfies the following conditions.

\begin{itemize}
	\item Every processed vertex of \(Y^I\) assigned value \(0\) has a
	processed neighbour in \(X\) assigned value \(2\).
	
	\item Every processed vertex of \(X\) assigned value \(0\) is either
	already dominated by a vertex of \(S\), dominated by a processed vertex
	of \(Y^I\) assigned value \(2\), or represented by the pending
	coordinate \(p\).
	
	\item The coordinates \(a\) and \(b\) are respectively the smallest and
	largest indices of the processed vertices of \(X\) assigned value \(2\).
\end{itemize}

Before any vertex is processed, the only feasible state is
\((a_\infty,0,p_\infty)\), with value \(0\). Therefore, we initialize
\(D(a_\infty,0,p_\infty)=0\) and set every other entry to \(+\infty\).

\paragraph*{\underline{\textbf{Why the coordinates \(a\) and \(b\) are needed}}}

The coordinate \(b\) is inherited from the dynamic program for convex
bipartite graphs. It records the largest index of a processed vertex of \(X\)
assigned value \(2\). When an interval vertex \(y\in Y_i^I\), with
\(N(y)=\{x_{l(y)},\ldots,x_i\}\), is assigned value \(0\), the choice is
feasible exactly when \(b\geq l(y)\). Indeed, because \(b\leq i\) at this
stage, the condition \(b\geq l(y)\) implies that \(x_b\in N(y)\) and
\(f(x_b)=2\). Thus, the coordinate \(b\) is necessary during the processing
of the ordinary interval vertices.

The coordinate \(a\), on the other hand, is introduced specifically to
handle the wrap-around vertices after the dynamic program terminates. It
records the smallest index of a vertex of \(X\) assigned the value \(2\). Let
\(y\in Y^W\) be a wrap-around vertex with
\(N(y)=\{x_{l(y)},\ldots,x_m\}\cup\{x_1,\ldots,x_{r(y)}\}\). A value-\(2\)
vertex lies in the prefix \(\{x_1,\ldots,x_{r(y)}\}\) if and only if
\(a\leq r(y)\). Similarly, a value-\(2\) vertex lies in the suffix
\(\{x_{l(y)},\ldots,x_m\}\) if and only if \(b\geq l(y)\). Consequently,
\(y\) is dominated by a value-\(2\) vertex of \(X\) if and only if
\(a\leq r(y)\) or \(b\geq l(y)\).

Therefore, \(b\) has two roles: it is used while processing the interval
vertices in \(Y^I\), and it is also used at the end to test the suffix part
of each wrap-around neighbourhood. The coordinate \(a\) is needed only for
testing the prefix part of a wrap-around neighbourhood. Together, \(a\) and
\(b\) are sufficient, and the exact set of vertices of \(X\) assigned value
\(2\) does not need to be stored.

\paragraph*{\underline{\textbf{Transitions for a vertex of \texorpdfstring{\(X\)}}}}

Suppose that the current state is \((a,b,p)\), and that \(x_i\) is the next
vertex to be processed. We consider the three possible values of \(f(x_i)\).

\begin{itemize}
	\item Suppose that \(f(x_i)=0\). If
	\(\operatorname{pre}_S(x_i)=0\), then \(x_i\) is already dominated by a
	vertex of \(S\), so no new pending obligation is created, and the new
	pending coordinate remains \(p\). If
	\(\operatorname{pre}_S(x_i)=1\), then \(x_i\) is not dominated by \(S\).
	It must therefore become pending, and the new pending coordinate is
	\(\min\{p,i\}\). Since \(x_i\) is the currently processed vertex, this
	value equals \(i\) whenever \(p=p_\infty\), and it equals \(p\) whenever
	an earlier pending vertex already exists. This transition adds no weight.
	
	\item Suppose that \(f(x_i)=1\). A vertex assigned value \(1\) has no
	domination requirement and does not dominate any other vertex for the
	purpose of Roman domination. Hence, the state remains \((a,b,p)\), and
	the weight increases by \(1\).
	
	\item Suppose that \(f(x_i)=2\). Since the vertices of \(X\) are processed
	in increasing order, \(x_i\) becomes the rightmost processed value-\(2\)
	vertex of \(X\). Therefore, the first two new coordinates are
	\(\min\{a,i\}\) and \(i\), respectively. The pending coordinate remains
	\(p\), because there are no edges between vertices of \(X\); assigning
	value \(2\) to \(x_i\) cannot dominate a pending vertex of \(X\).
	Accordingly, the new state is \((\min\{a,i\},i,p)\), and the weight
	increases by \(2\).
\end{itemize}

For each transition, if more than one partial assignment produces the same
new state, we retain only the one of minimum weight.

\paragraph*{\underline{\textbf{Transitions for an interval vertex of \(Y^I\)}}}

After processing \(x_i\), the vertices of \(Y_i^I\) are processed one at a
time. Let \(y\in Y_i^I\). Since \(r(y)=i\), its neighbourhood is
\(N(y)=\{x_{l(y)},\ldots,x_i\}\). Again, we consider the three possible
values of \(f(y)\).

\begin{itemize}
	\item Suppose that \(f(y)=0\). Then \(y\) must have a neighbour in \(X\)
	assigned value \(2\). Since every neighbour of \(y\) has already been
	processed, this condition can be checked immediately. The rightmost
	processed value-\(2\) vertex of \(X\) is \(x_b\). Therefore, a
	value-\(2\) vertex lies in \(N(y)\) exactly when \(b\geq l(y)\). If \(b\geq l(y)\), the transition is feasible and the state remains \((a,b,p)\); also, the weight does not increase. If \(b<l(y)\), then no value-\(2\) vertex of \(X\) lies in \(N(y)\), so the choice \(f(y)=0\) is infeasible.

	\item Suppose that \(f(y)=1\). The vertex \(y\) is safe by itself, and it
	does not remove any pending obligation. Therefore, the state remains
	\((a,b,p)\), and the weight increases by \(1\).
	
	\item Suppose that \(f(y)=2\). The vertex \(y\) dominates every pending
	vertex in \(N(y)=\{x_{l(y)},\ldots,x_i\}\). If \(p=p_\infty\), then there
	is no pending vertex, so the pending coordinate remains \(p_\infty\). If
	\(p\neq p_\infty\) and \(p\geq l(y)\), then \(x_p\in N(y)\). Since \(p\)
	is the earliest pending index, every other pending vertex has index at
	least \(p\) and at most \(i\), and hence also lies in \(N(y)\).
	Therefore, all pending vertices become dominated and the new pending
	coordinate is \(p_\infty\). If \(p<l(y)\), then \(y\) does not dominate
	the earliest pending vertex, so the pending coordinate remains \(p\).
	The coordinates \(a\) and \(b\) do not change, because they record only
	value-\(2\) vertices belonging to \(X\). This transition increases the
	weight by \(2\).
\end{itemize}

An interval vertex \(y\in Y_i^I\) cannot be dominated by a vertex of \(S\),
because both \(y\) and every vertex of \(S\) belong to the bipartition class
\(Y\). Similarly, a value-\(2\) interval vertex cannot dominate another
vertex of \(Y\). This is why the feasibility of the choice \(f(y)=0\) depends
only on the value-\(2\) vertices of \(X\), as recorded by \(b\).

\paragraph*{\underline{\textbf{Terminal states}}}

After processing \(x_m\) and every vertex of \(Y_m^I\), all vertices of
\(X\cup Y^I\) have been assigned values. We retain only states of the form
\((a,b,p_\infty)\). The condition \(p=p_\infty\) guarantees that no vertex of
\(X\) assigned value \(0\) remains undominated. Hence, every value-\(0\)
vertex of \(X\) is dominated either by a forced value-\(2\) vertex of \(S\)
or by a value-\(2\) interval vertex of \(Y^I\). Every value-\(0\) vertex of
\(Y^I\) was checked for domination at the time it was processed.

\paragraph*{\underline{\textbf{Completing the assignment on \(Y^W\setminus S\)}}}

It remains to assign values to the wrap-around vertices outside \(S\). By the
definition of the current branch, no vertex of \(Y^W\setminus S\) may receive
value \(2\). Consider a final state \((a,b,p_\infty)\) and a vertex
\(y\in Y^W\setminus S\). The vertex \(y\) can receive value \(0\) if and only
if it has a value-\(2\) neighbour in \(X\), which occurs exactly when
\(a\leq r(y)\) or \(b\geq l(y)\). If neither condition holds, then \(y\)
cannot receive value \(0\) and must receive value \(1\).

The choices for the vertices of \(Y^W\setminus S\) are independent once
\(a\) and \(b\) are fixed, because there are no edges inside \(Y\).
Therefore, let \(c(S,a,b)\) be the number of vertices
\(y\in Y^W\setminus S\) satisfying \(a>r(y)\) and \(b<l(y)\). Each vertex counted by \(c(S,a,b)\) receives value \(1\), while every other
vertex of \(Y^W\setminus S\) receives value \(0\).

Consequently, for the branch corresponding to \(S\), we minimize
\(D(a,b,p_\infty)+c(S,a,b)\) over all reachable final states
\((a,b,p_\infty)\), and then add the fixed contribution \(2|S|\). Finally,
the algorithm takes the minimum value over all subsets \(S\subseteq Y^W\)
with \(|S|\leq 2\).

\subsection{Algorithm}

We divide the algorithm into three parts.
Algorithm~\ref{alg:roman-circular-convex} is the main routine,
Algorithm~\ref{alg:process-circular-x} processes a vertex of \(X\), and
Algorithm~\ref{alg:process-circular-y} processes an interval vertex of
\(Y^I\).

As in the convex-bipartite algorithm, \(a_\infty\) and \(p_\infty\) are
special state values, whereas \(+\infty\) denotes an infeasible table entry.
The value \(a_\infty\) means that no processed vertex of \(X\) has been
assigned value \(2\), while \(p_\infty\) means that no processed vertex of
\(X\) is currently pending. Both state values are represented numerically by
\(m+1\).

For a fixed set \(S\subseteq Y^W\), the main routine initializes the dynamic
programming table and processes the vertices in the order
\(x_1,Y_1^I,x_2,Y_2^I,\ldots,x_m,Y_m^I\). The subroutine
\textsc{Process-X} applies the three transitions corresponding to assigning
value \(0\), \(1\), or \(2\) to a vertex \(x_i\), whereas the subroutine
\textsc{Process-Interval-Y} applies the corresponding transitions to an
interval vertex \(y\in Y^I\).

A state is called \emph{finite} if its table entry is smaller than
\(+\infty\); equivalently, the state is realized by at least one feasible
partial assignment. Whenever a vertex is processed, the corresponding
subroutine initializes every entry of a temporary table \(D'\) to
\(+\infty\). The instruction ``update \(D'(a,b,p)\) with \(w\)'' means that
\(D'(a,b,p)\) is replaced by \(\min\{D'(a,b,p),w\}\). Hence, if several
transitions produce the same state, only the minimum resulting weight is
retained. After all transitions for the current vertex have been considered,
the updated table \(D'\) is returned to the calling routine and replaces the
current table \(D\).

After all vertices of \(X\cup Y^I\) have been processed, the main routine
considers only finite final states of the form \((a,b,p_\infty)\). For each
such state, it counts the vertices of \(Y^W\setminus S\) that are not
dominated by any value-\(2\) vertex of \(X\). These vertices must receive
value \(1\), and their number is added to the weight of the corresponding
final state.

\begin{algorithm}[H]
	\small
	\DontPrintSemicolon
	\caption{\textsc{Roman-Circular-Convex-Bipartite}}
	\label{alg:roman-circular-convex}
	\KwIn{A circular-convex bipartite graph \(G=(X\cup Y,E)\) with a
		circular-convex representation on
		\(X=\{x_1,\ldots,x_m\}\).}
	\KwOut{The Roman domination number \(\gamma_R(G)\).}
	
	Cut the circular ordering between \(x_m\) and \(x_1\)\;
	Partition \(Y\) into \(Y^I\) and \(Y^W\)\;
	
	\For{\(i=1\) \KwTo \(m\)}{
		Set \(Y_i^I=\{y\in Y^I:r(y)=i\}\)\;
	}
	
	Let \(a_\infty=p_\infty=m+1\)\;
	Set \(\gamma_R=+\infty\)\;
	
	\ForEach{\(S\subseteq Y^W\) with \(|S|\leq 2\)}{
		\ForEach{\(x_i\in X\)}{
			\eIf{\(N(x_i)\cap S\neq\emptyset\)}{
				Set \(\operatorname{pre}_S(x_i)=0\)\;
			}{
				Set \(\operatorname{pre}_S(x_i)=1\)\;
			}
		}
		
		Set all entries of \(D\) to \(+\infty\)\;
		Set \(D(a_\infty,0,p_\infty)=0\)\;
		
		\For{\(i=1\) \KwTo \(m\)}{
			Set
			\(D=\textsc{Process-X}
			(D,i,\operatorname{pre}_S(x_i))\)\;
			
			\ForEach{\(y\in Y_i^I\), in an arbitrary fixed order}{
				Set
				\(D=\textsc{Process-Interval-Y}(D,y)\)\;
			}
		}
		
		\ForEach{finite final state \((a,b,p_\infty)\) of \(D\)}{
			Set
			\(c=
			|\{y\in Y^W\setminus S:
			a>r(y)\text{ and }b<l(y)\}|\)\;
			
			Set
			\(\gamma_R=
			\min\{\gamma_R,
			D(a,b,p_\infty)+2|S|+c\}\)\;
		}
	}
	
	\Return{\(\gamma_R\)}\;
\end{algorithm}

\begin{algorithm}[H]
	\small
	\DontPrintSemicolon
	\caption{\textsc{Process-X}}
	\label{alg:process-circular-x}
	\KwIn{The current table \(D\), an index \(i\), and
		\(\operatorname{pre}_S(x_i)\).}
	\KwOut{The updated table after processing \(x_i\).}
	
	Set all entries of \(D'\) to \(+\infty\)\;
	
	\ForEach{finite state \((a,b,p)\) of \(D\)}{
		\eIf{\(\operatorname{pre}_S(x_i)=0\)}{
			Set \(p_0=p\)\;
		}{
			Set \(p_0=\min\{p,i\}\)\;
		}
		
		Update \(D'(a,b,p_0)\) with \(D(a,b,p)\)
		\tcp*{\(f(x_i)=0\)}
		
		Update \(D'(a,b,p)\) with \(D(a,b,p)+1\)
		\tcp*{\(f(x_i)=1\)}
		
		Set \(a_2=\min\{a,i\}\)\;
		Update \(D'(a_2,i,p)\) with \(D(a,b,p)+2\)
		\tcp*{\(f(x_i)=2\)}
	}
	
	\Return{\(D'\)}\;
\end{algorithm}

\begin{algorithm}[H]
	\small
	\DontPrintSemicolon
	\caption{\textsc{Process-Interval-Y}}
	\label{alg:process-circular-y}
	\KwIn{The current table \(D\) and an interval vertex \(y\in Y^I\).}
	\KwOut{The updated table after processing \(y\).}
	
	Set \(l=l(y)\)\;
	Set all entries of \(D'\) to \(+\infty\)\;
	
	\ForEach{finite state \((a,b,p)\) of \(D\)}{
		\If{\(b\geq l\)}{
			Update \(D'(a,b,p)\) with \(D(a,b,p)\)
			\tcp*{\(f(y)=0\)}
		}
		
		Update \(D'(a,b,p)\) with \(D(a,b,p)+1\)
		\tcp*{\(f(y)=1\)}
		
		\eIf{\(p\neq p_\infty\) and \(p\geq l\)}{
			Set \(p_2=p_\infty\)\;
		}{
			Set \(p_2=p\)\;
		}
		
		Update \(D'(a,b,p_2)\) with \(D(a,b,p)+2\)
		\tcp*{\(f(y)=2\)}
	}
	
	\Return{\(D'\)}\;
\end{algorithm}

\subsection{Correctness Proof}

We prove that Algorithm~\ref{alg:roman-circular-convex} computes
\(\gamma_R(G)\). The proof has two parts. First, we show that, for every fixed
branch \(S\), the dynamic program correctly computes the best assignment on
\(X\cup Y^I\) consistent with forcing the vertices of \(S\) to receive value
\(2\). We then show that branching over all sets \(S\subseteq Y^W\) with
\(|S|\leq 2\) is sufficient.

Fix a subset \(S\subseteq Y^W\) with \(|S|\leq 2\). In this branch, every
vertex of \(S\) is forced to receive value \(2\). Hence, the vertices of \(X\)
dominated by \(S\) are known before the dynamic program starts, and this
information is stored by \(\operatorname{pre}_S\).

We first consider the dynamic program on \(X\cup Y^I\). If a vertex \(x_i\)
is assigned value \(0\) and \(\operatorname{pre}_S(x_i)=0\), then \(x_i\) is
already dominated by a vertex of \(S\), so it does not become pending. If
\(\operatorname{pre}_S(x_i)=1\), then \(x_i\) is not pre-dominated by \(S\)
and therefore, becomes pending exactly as in the convex-bipartite algorithm.
The interval vertices in \(Y^I\) are processed as in the convex-bipartite
algorithm of Das et al.~\cite{das2026roman}, because their neighbourhoods are
ordinary intervals in the linear order \(x_1,\ldots,x_m\).

The state \((a,b,p)\) stores all the information needed for the remainder of
the branch. The value \(p\) records the earliest processed vertex of \(X\)
that is still pending. This is sufficient for the same reason as in the
convex case: every future interval vertex sees the already processed part of
\(X\) as a suffix. In particular, if a future interval vertex assigned value
\(2\) reaches the earliest pending vertex, then it also reaches every later
pending vertex. The values \(a\) and \(b\) record the leftmost and rightmost
processed vertices of \(X\) assigned value \(2\). These values are sufficient
to decide whether a wrap-around vertex is dominated by a vertex of \(X\)
assigned value \(2\). Indeed, a wrap-around vertex \(y\) with
\(N(y)=\{x_{l(y)},\ldots,x_m\}\cup\{x_1,\ldots,x_{r(y)}\}\) is dominated by
an \(X\)-vertex assigned value \(2\) if and only if
\(a\leq r(y)\) or \(b\geq l(y)\).

An induction on the processing order proves the following invariant. After
each processing step, if there exists a feasible partial assignment with
state \((a,b,p)\), then \(D(a,b,p)\) is the minimum weight among all such
partial assignments. If no feasible partial assignment realizes
\((a,b,p)\), then \(D(a,b,p)=+\infty\).

The invariant holds initially because
\(D(a_\infty,0,p_\infty)=0\) represents the empty assignment and all other states
are infeasible. When \(x_i\) is processed, the three transitions consider all
possible values of \(f(x_i)\) and correctly determine whether \(x_i\) becomes
pending. When an interval vertex \(y\in Y_i^I\) is processed, the three
transitions consider all possible values of \(f(y)\). The choice \(f(y)=0\)
is allowed exactly when \(y\) has a value-\(2\) neighbour in \(X\), while the
choice \(f(y)=2\) removes precisely the pending vertices contained in
\(N(y)\). Therefore, every feasible extension is considered, and every
generated extension is feasible. Taking the minimum over all transitions
preserves the invariant.

After the dynamic program finishes, we accept only states with
\(p=p_\infty\), because no vertex of \(X\) may remain pending in a complete
Roman dominating function. It remains to assign values to the vertices of
\(Y^W\setminus S\). In the branch corresponding to \(S\), these vertices are
not allowed to receive value \(2\). A vertex \(y\in Y^W\setminus S\) can
receive value \(0\) exactly when it is dominated by some vertex of \(X\)
assigned value \(2\). Since \(y\) is a wrap-around vertex, this happens
exactly when \(a\leq r(y)\) or \(b\geq l(y)\). If this condition fails, then
\(y\) must receive value \(1\). Therefore, \(c(S,a,b)\) is exactly the
minimum additional cost required for the wrap-around vertices not in \(S\).

It remains to show that branching over sets \(S\) with \(|S|\leq 2\) is
sufficient. By Lemma~\ref{lem:circular-two-wrap}, there exists an optimal
Roman dominating function in which at most two wrap-around vertices are
assigned value \(2\). Let \(S\) be precisely this set of wrap-around
vertices. The algorithm considers this choice of \(S\). For this branch, the
dynamic program computes the best possible assignment on \(X\cup Y^I\), and
the final counting step assigns the best possible values to the remaining
wrap-around vertices. Hence, the algorithm obtains a solution of weight at
most \(\gamma_R(G)\).

Conversely, every solution constructed by the algorithm is a valid Roman
dominating function of \(G\). The vertices in \(S\) receive value \(2\). The
vertices in \(X\cup Y^I\) are handled by the dynamic program, and no vertex
of \(X\) remains pending in an accepted final state. Finally, every vertex of
\(Y^W\setminus S\) either receives value \(0\) because it is dominated by an
\(X\)-vertex assigned value \(2\), or receives value \(1\). Therefore, the
algorithm never underestimates \(\gamma_R(G)\).

Thus, the minimum value returned over all choices of \(S\) is exactly
\(\gamma_R(G)\).

\subsection{Running Time Analysis}

\begin{theorem}
	\label{thm:circular-convex-rdp}
	Given a circular-convex representation of a circular-convex bipartite graph
	\(G\) on \(n\) vertices, Algorithm~\ref{alg:roman-circular-convex} computes
	\(\gamma_R(G)\) in \(O(n^6)\) time.
\end{theorem}

\begin{proof}
	Let \(m=|X|\) and \(q=|Y|\). There are at most
	\(1+|Y^W|+\binom{|Y^W|}{2}=O(n^2)\) choices for the set \(S\).
	
	For a fixed set \(S\), the dynamic program has states of the form
	\((a,b,p)\). The coordinate \(a\) has \(m+1\) possible values, namely
	\(\{1,\ldots,m\}\cup\{a_\infty\}\). The coordinate \(b\) has \(m+1\) possible
	values, namely \(\{0,1,\ldots,m\}\). The coordinate \(p\) has \(m+1\)
	possible values, namely \(\{1,\ldots,m\}\cup\{p_\infty\}\). Hence, the number
	of states is \(O(m^3)\).
	
	Each vertex of \(X\cup Y^I\) is processed once, and for every reachable state
	the algorithm considers a constant number of transitions. Each transition
	takes constant time. Therefore, for one fixed set \(S\), the dynamic
	programming part takes \(O((m+|Y^I|)m^3)\) time, which is \(O(nm^3)\).
	
	For each reachable final state \((a,b,p_\infty)\), the additional cost for the
	vertices of \(Y^W\setminus S\) can be computed directly in
	\(O(|Y^W|)\) time. There are at most \(O(m^2)\) final states with
	\(p=p_\infty\). Thus, the final counting step takes
	\(O(|Y^W|m^2)\) time for one fixed set \(S\), which is dominated by
	\(O(nm^3)\). Therefore, one branch takes \(O(nm^3)\) time.
	
	Since there are \(O(n^2)\) choices for \(S\) and \(m\leq n\), the total
	running time is \(O(n^2\cdot nm^3)=O(n^6)\).
\end{proof}

\begin{remark}
	The assumption that a circular-convex representation is given is important.
	The algorithm uses the circular ordering of \(X\) and the endpoints \(l(y)\)
	and \(r(y)\) of each circular interval. Recognition and construction of such
	a representation are separate preprocessing issues and are not part of the
	dynamic program described here.
\end{remark}

\subsection{Example}

We illustrate Algorithm~\textsc{Roman-Circular-Convex-Bipartite} on a small
circular-convex bipartite graph. Let
\(X=\{x_1,x_2,x_3,x_4,x_5\}\), placed in the circular order
\(x_1,x_2,x_3,x_4,x_5,x_1\), and let \(Y=\{y_1,y_2,y_3\}\). Suppose that
\(N(y_1)=\{x_5,x_1,x_2\}\),
\(N(y_2)=\{x_4,x_5,x_1\}\), and
\(N(y_3)=\{x_2,x_3,x_4\}\). Thus, \(y_1\) and \(y_2\) are wrap-around
vertices, while \(y_3\) is an interval vertex after cutting the circle
between \(x_5\) and \(x_1\). More precisely, after the cut we have
\(l(y_1)=5\), \(r(y_1)=2\), \(l(y_2)=4\), \(r(y_2)=1\),
\(l(y_3)=2\), and \(r(y_3)=4\). Hence,
\(Y^W=\{y_1,y_2\}\), \(Y^I=\{y_3\}\), and \(Y_4^I=\{y_3\}\), while all
other sets \(Y_i^I\) are empty. The graph and the Roman dominating function
obtained by the algorithm are shown in
Figure~\ref{fig:roman-circular-convex-example}.

\begin{figure}[h]
	\centering
		\begin{tikzpicture}[
				scale=1.10,
			xnode/.style={
				circle,
				draw,
				minimum size=7mm,
				inner sep=0pt
			},
			ynode/.style={
				circle,
				draw,
				minimum size=7mm,
				inner sep=0pt
			},
			opt2/.style={
				circle,
				draw,
				fill=gray!30,
				very thick,
				minimum size=7mm,
				inner sep=0pt
			},
			one/.style={
				circle,
				draw,
				minimum size=7mm,
				inner sep=0pt,
				very thick
			},
			every node/.style={font=\small}
			]
			\node[xnode] (x1) at (0,2.0) {\(x_1\)};
			\node[xnode] (x2) at (1.5,0.6) {\(x_2\)};
			\node[one]   (x3) at (1.0,-1.6) {\(x_3\)};
			\node[opt2]  (x4) at (-1.0,-1.6) {\(x_4\)};
			\node[xnode] (x5) at (-1.5,0.6) {\(x_5\)};
			
			\node[opt2]  (y1) at (0,3.5) {\(y_1\)};
			\node[ynode] (y2) at (-3.4,-0.1) {\(y_2\)};
			\node[ynode] (y3) at (3.4,-0.1) {\(y_3\)};
			
			\draw[dashed] (x1)--(x2)--(x3)--(x4)--(x5)--(x1);
			
			\draw (y1)--(x5);
			\draw (y1)--(x1);
			\draw (y1)--(x2);
			
			\draw (y2)--(x4);
			\draw (y2)--(x5);
			\draw (y2)--(x1);
			
			\draw (y3)--(x2);
			\draw (y3)--(x3);
			\draw (y3)--(x4);
			
			\node at (0,-2.35) {circular order on \(X\)};
			\node[right] at (1.55,-1.6) {\(1\)};
			\node[left] at (-1.55,-1.6) {\(2\)};
			\node[above] at (0,3.85) {\(2\)};
		\end{tikzpicture}%
	
	\caption{A circular-convex bipartite graph used in the example.
		The shaded vertices receive value \(2\), namely \(f(y_1)=2\) and
		\(f(x_4)=2\), and \(x_3\) receives value \(1\). All other vertices
		receive value \(0\).}
	\label{fig:roman-circular-convex-example}
\end{figure}
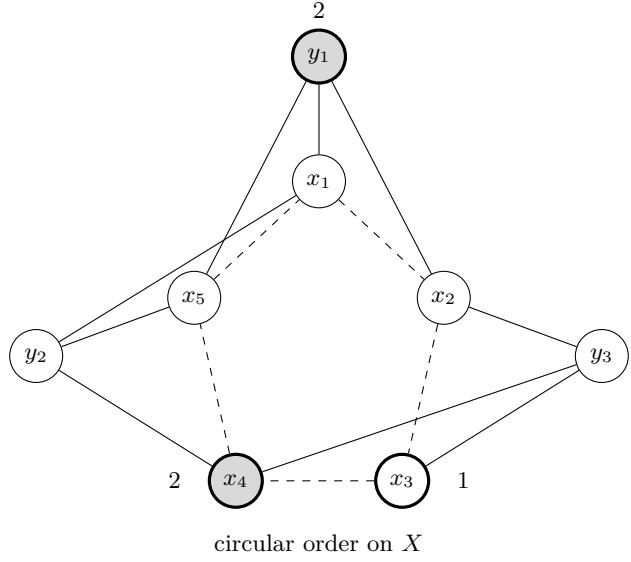

We now trace the optimal branch illustrated in
Figure~\ref{fig:roman-circular-convex-example}. Consider the branch
\(S=\{y_1\}\). Thus, \(y_1\) is forced to receive value \(2\), and it
contributes \(2\) to the total weight. Since
\(N(y_1)=\{x_5,x_1,x_2\}\), the pre-dominated vertices of \(X\) are
\(x_1\), \(x_2\), and \(x_5\). Therefore,
\(\operatorname{pre}_S(x_1)=
\operatorname{pre}_S(x_2)=
\operatorname{pre}_S(x_5)=0\), while
\(\operatorname{pre}_S(x_3)=
\operatorname{pre}_S(x_4)=1\).

For this fixed branch, the dynamic program is run on \(X\cup Y^I\). Since
\(Y^I=\{y_3\}\) and \(r(y_3)=4\), the processing order is
\(x_1,x_2,x_3,x_4,y_3,x_5\). The state is \((a,b,p)\), where \(a\) is the
leftmost processed vertex of \(X\) assigned value \(2\), \(b\) is the
rightmost processed vertex of \(X\) assigned value \(2\), and \(p\) is the
earliest pending vertex of \(X\). Initially, we have
\(D(a_\infty,0,p_\infty)=0\).

The following table traces one optimal branch of the dynamic program for
\(S=\{y_1\}\).

\begin{center}
	\large
	\begin{tabular}{c|c|c}
		Step & Chosen value & Resulting state and value \\
		\hline
		Initial
		& --
		& \(D(a_\infty,0,p_\infty)=0\) \\
		
		After \(x_1\)
		& \(f(x_1)=0\), \(\operatorname{pre}_S(x_1)=0\)
		& \(D(a_\infty,0,p_\infty)=0\) \\
		
		After \(x_2\)
		& \(f(x_2)=0\), \(\operatorname{pre}_S(x_2)=0\)
		& \(D(a_\infty,0,p_\infty)=0\) \\
		
		After \(x_3\)
		& \(f(x_3)=1\)
		& \(D(a_\infty,0,p_\infty)=1\) \\
		
		After \(x_4\)
		& \(f(x_4)=2\)
		& \(D(4,4,p_\infty)=3\) \\
		
		After \(y_3\)
		& \(f(y_3)=0\)
		& \(D(4,4,p_\infty)=3\) \\
		
		After \(x_5\)
		& \(f(x_5)=0\), \(\operatorname{pre}_S(x_5)=0\)
		& \(D(4,4,p_\infty)=3\)
	\end{tabular}
\end{center}

We explain these transitions in more detail. Since \(x_1\) and \(x_2\) are
already dominated by \(y_1\), assigning value \(0\) to them does not create a
pending vertex. Thus, the state remains \((a_\infty,0,p_\infty)\) after
processing \(x_1\) and \(x_2\). The vertex \(x_3\) is not pre-dominated by
\(S\), but we assign \(f(x_3)=1\). Hence, it is safe by itself, and the state
still has no pending vertex; the weight becomes \(1\).

Next, we assign \(f(x_4)=2\). This makes \(x_4\) both the leftmost and the
rightmost processed vertex of \(X\) assigned value \(2\). Therefore, the
first two coordinates become \(a=4\) and \(b=4\), and the current table entry
is \(D(4,4,p_\infty)=3\).

Now the interval vertex \(y_3\) is processed. Since
\(N(y_3)=\{x_2,x_3,x_4\}\), we have \(l(y_3)=2\). The current rightmost
value-\(2\) vertex of \(X\) is \(x_4\), so \(b=4\geq l(y_3)\). Hence,
\(y_3\) is dominated by \(x_4\), and assigning \(f(y_3)=0\) is allowed. The
state remains \((4,4,p_\infty)\), and the weight remains \(3\). Finally,
\(x_5\) is already dominated by \(y_1\), since
\(\operatorname{pre}_S(x_5)=0\). Therefore, assigning \(f(x_5)=0\) does not
create a pending vertex, and the state remains \((4,4,p_\infty)\).

After the dynamic program finishes for this branch, the final state
\((4,4,p_\infty)\) is accepted because \(p=p_\infty\). We now handle the
wrap-around vertices in \(Y^W\setminus S\). Here,
\(Y^W\setminus S=\{y_2\}\). The vertex \(y_2\) has
\(N(y_2)=\{x_4,x_5,x_1\}\), so \(l(y_2)=4\) and \(r(y_2)=1\). Since the
final state has \(a=4\) and \(b=4\), the condition
\(a\leq r(y_2)\) or \(b\geq l(y_2)\) becomes
\(4\leq 1\) or \(4\geq 4\). The second inequality holds, so \(y_2\) is
dominated by \(x_4\). Hence, \(y_2\) can receive value \(0\), and the
additional cost is \(c(S,4,4)=0\).

The total weight obtained in this branch is
\(D(4,4,p_\infty)+2|S|+c(S,4,4)=3+2+0=5\). Thus, the algorithm finds the Roman
dominating function shown in
Figure~\ref{fig:roman-circular-convex-example}, namely
\(f(y_1)=2\), \(f(x_4)=2\), \(f(x_3)=1\), and all other vertices receive
value \(0\).

This value is optimal. Suppose, for a contradiction, that there exists a
Roman dominating function \(f\) of weight at most \(4\). First, \(f\) cannot
assign value \(2\) to no vertex. Indeed, in that case no vertex could receive
value \(0\), because every vertex assigned value \(0\) must have a neighbour
assigned value \(2\). Hence, all eight vertices would have to receive value
\(1\), giving weight \(8\).

Suppose next that exactly one vertex is assigned value \(2\). If this vertex
lies in \(X\), then every other vertex of \(X\) must receive a positive value,
because there are no edges inside \(X\). Thus, the total weight is at least
\(2+4=6\).

If the unique value-\(2\) vertex lies in \(Y\), then the other two vertices
of \(Y\) must receive positive values, because there is no value-\(2\) vertex
in \(X\). Moreover, every vertex of \(X\) outside the neighbourhood of the
unique value-\(2\) vertex must receive a positive value. Since every vertex
of \(Y\) has exactly three neighbours in \(X\), two vertices of \(X\) lie
outside its neighbourhood. Consequently, the total weight is at least
\(2+2+2=6\). Therefore, a Roman dominating function of weight at most \(4\)
cannot have exactly one vertex assigned value \(2\).

It follows that a Roman dominating function of weight at most \(4\) would
have exactly two vertices assigned value \(2\), and every other vertex would
be assigned value \(0\). If both value-\(2\) vertices lie in \(X\), then the
remaining vertices of \(X\) are assigned value \(0\) but have no value-\(2\)
neighbour. If both value-\(2\) vertices lie in \(Y\), then the remaining
vertex of \(Y\) is assigned value \(0\) but has no value-\(2\) neighbour.

Therefore, one value-\(2\) vertex must lie in \(X\) and the other must lie in
\(Y\). However, each vertex of \(Y\) misses two vertices of \(X\). At most
one of these two missed vertices can be the selected value-\(2\) vertex of
\(X\). Hence, at least one vertex of \(X\) that is not adjacent to the
selected value-\(2\) vertex of \(Y\) is assigned value \(0\). This vertex
cannot be dominated by the selected value-\(2\) vertex of \(X\), because
there are no edges inside \(X\). This is a contradiction.

Therefore, no Roman dominating function of weight at most \(4\) exists, and
the optimum value is \(\gamma_R(G)=5\).

\section{Roman Domination on Triad-Convex Bipartite Graphs}
\label{sec:rd-triad}

In this section, we extend the dynamic programming approach for Roman
Domination on convex bipartite graphs due to Das et
al.~\cite{das2026roman} to triad-convex bipartite graphs. Let
\(G=(X\cup Y,E)\) be a bipartite graph. We assume that a triad-convex
representation of \(G\) is given on the partite set \(X\). Thus, there is a
tree \(T\) with vertex set \(X\), where \(T\) is a subdivision of
\(K_{1,3}\), and, for every vertex \(y\in Y\), the set \(N(y)\) induces a
connected subtree of \(T\). Triad-convex bipartite graphs form a restricted
subclass of tree-convex bipartite
graphs.

Let \(c\) be the unique branching vertex of the triad. Removing \(c\) from
\(T\) gives three paths, called the \emph{arms} of the triad. We denote them
by \(P_1,P_2,P_3\), where
\(P_k=\{x_{k,1},x_{k,2},\ldots,x_{k,m_k}\}\) is ordered from the center
outward. Thus, \(x_{k,1}\) is the vertex closest to \(c\) on the \(k\)-th
arm.

A vertex \(y\in Y\) is called an \emph{arm vertex} if \(N(y)\) is contained
in one arm and does not contain \(c\). In this case, \(N(y)\) is an interval
on that arm. A vertex \(y\in Y\) is called a \emph{central \(Y\)-vertex} if
\(c\in N(y)\). Since \(N(y)\) is connected in the triad, every central
\(Y\)-vertex has neighbourhood
\(N(y)=\{c\}\cup\bigcup_{k=1}^{3}
\{x_{k,1},x_{k,2},\ldots,x_{k,r_k(y)}\}\), where
\(r_k(y)\in\{0,1,\ldots,m_k\}\). If \(r_k(y)=0\), then \(y\) has no
neighbour on the \(k\)-th arm. Let \(Y^C\) be the set of central
\(Y\)-vertices, and let \(Y_k\) be the set of arm vertices whose
neighbourhoods lie on \(P_k\).

For each \(k\in\{1,2,3\}\), the subgraph induced by \(P_k\cup Y_k\) is a
convex bipartite graph with respect to the linear order
\(x_{k,1},x_{k,2},\ldots,x_{k,m_k}\). Therefore, the convex-bipartite
dynamic programming algorithm of Das et al.~\cite{das2026roman} can be used
as the base routine on each arm. The difficulty comes from the central
\(Y\)-vertices, because they may interact with more than one arm. We handle
them by branching over the central \(Y\)-vertices that receive Roman value
\(2\).

\subsection{A Structural Observation}

We first show that, in some optimal Roman dominating function, only a
constant number of central \(Y\)-vertices need to receive value \(2\).

\begin{lemma}
	\label{lem:triad-three-central}
	There exists an optimal Roman dominating function in which at most three
	vertices of \(Y^C\) are assigned value \(2\).
\end{lemma}

\begin{proof}
	Let \(f\) be an optimal Roman dominating function, and let
	\(S_f=\{y\in Y^C:f(y)=2\}\). If \(|S_f|\leq 3\), then there is nothing to
	prove. Suppose that \(|S_f|>3\).
	
	For each arm \(P_k\), choose a vertex \(y_k\in S_f\) such that \(r_k(y_k)\)
	is maximum among all vertices of \(S_f\). If this maximum value is \(0\),
	then no vertex of \(S_f\) dominates a vertex on \(P_k\), and no
	representative is needed for that arm. Let \(S'\) be the set of
	representatives chosen in this way. If \(S'=\emptyset\), then every vertex
	of \(S_f\) has neighbourhood exactly \(\{c\}\); in this case, add one
	arbitrary vertex of \(S_f\) to \(S'\). In all cases, \(|S'|\leq 3\).
	
	We claim that the vertices of \(S'\) dominate every vertex of \(X\) that is
	dominated by the vertices of \(S_f\). The center \(c\) is dominated by every
	vertex of \(S'\), because every vertex of \(S'\) is central. If \(S'\)
	consists of one arbitrary vertex selected because all arm reaches are zero,
	then \(c\) is the only vertex of \(X\) dominated by \(S_f\). Otherwise,
	consider an arm \(P_k\). If some vertex of \(S_f\) dominates a vertex on
	\(P_k\), then the representative \(y_k\) was selected for that arm. By the
	choice of \(y_k\), the value \(r_k(y_k)\) is at least the reach of every
	vertex of \(S_f\) on \(P_k\). Hence, \(y_k\) dominates every vertex of
	\(P_k\) that was dominated by a vertex of \(S_f\). Therefore, the vertices
	of \(S'\) dominate all vertices of \(X\) that were dominated by \(S_f\).
	
	Now change the value of every vertex in \(S_f\setminus S'\) from \(2\) to
	\(1\). These vertices remain safe, because a vertex assigned value \(1\)
	does not need to be dominated. No vertex of \(X\) loses domination, because
	\(S'\) dominates every vertex of \(X\) that was dominated by \(S_f\).
	Moreover, vertices of \(Y\) do not dominate other vertices of \(Y\), since
	\(G\) is bipartite with bipartition \((X,Y)\). Thus, the modified function
	is still a Roman dominating function, but it has smaller weight,
	contradicting the optimality of \(f\). Therefore, some optimal Roman
	dominating function assigns value \(2\) to at most three central
	\(Y\)-vertices.
\end{proof}

\subsection{Arm Dynamic Program for a Fixed Central Choice}

By Lemma~\ref{lem:triad-three-central}, there exists an optimal Roman
dominating function in which at most three vertices of \(Y^C\) receive value
\(2\). We therefore, consider every subset \(S\subseteq Y^C\) with
\(|S|\leq 3\). In the branch corresponding to \(S\), every vertex of \(S\)
is forced to receive value \(2\), while every vertex of \(Y^C\setminus S\)
is forbidden from receiving value \(2\). We also branch over the value
assigned to the center \(c\). Let \(f(c)=\alpha\), where
\(\alpha\in\{0,1,2\}\).

Fix a set \(S\subseteq Y^C\) with \(|S|\leq 3\) and a value
\(\alpha\in\{0,1,2\}\). If \(\alpha=0\) and \(S=\emptyset\), then the branch
is infeasible. Indeed, the only neighbours of \(c\) belong to \(Y^C\), and
no such neighbour receives value \(2\) in this branch. In every other case,
the fixed vertices contribute \(2|S|+\alpha\) to the weight. This
contribution is not included in the arm tables and is added only when the
three arm solutions are combined.

\paragraph*{\underline{\textbf{Pre-dominated vertices on an arm}}}

The vertices of \(S\) may already dominate initial portions of the arms. For
each \(k\in\{1,2,3\}\) and each vertex \(x_{k,i}\in P_k\), define
\(\operatorname{pre}_S(x_{k,i})=0\) if \(x_{k,i}\) is dominated by a vertex
of \(S\), and define \(\operatorname{pre}_S(x_{k,i})=1\) otherwise. Thus,
\(\operatorname{pre}_S(x_{k,i})=0\) exactly when there exists a vertex
\(y\in S\) such that \(i\leq r_k(y)\).

Under this convention, the value \(0\) means that no new domination
obligation is needed if \(x_{k,i}\) is assigned value \(0\). By contrast,
\(\operatorname{pre}_S(x_{k,i})=1\) means that \(x_{k,i}\) is not dominated
from the central part. If such a vertex receives value \(0\), then it must
eventually be dominated by an arm vertex of \(Y_k\) assigned value \(2\).

\paragraph*{\underline{\textbf{Processing order on an arm}}}

Fix an arm \(P_k\). Every vertex \(y\in Y_k\) has an interval neighbourhood
on \(P_k\), and we write
\(N(y)=\{x_{k,l(y)},x_{k,l(y)+1},\ldots,x_{k,r(y)}\}\). For each
\(i\in\{1,\ldots,m_k\}\), let
\(Y_{k,i}=\{y\in Y_k:r(y)=i\}\). The vertices of \(P_k\cup Y_k\) are
processed in the order
\(x_{k,1},Y_{k,1},x_{k,2},Y_{k,2},\ldots,x_{k,m_k},Y_{k,m_k}\).
The vertices inside each set \(Y_{k,i}\) may be processed in any fixed order.

When a vertex \(y\in Y_{k,i}\) is processed, every vertex of \(N(y)\) has
already been processed, because
\(N(y)=\{x_{k,l(y)},\ldots,x_{k,i}\}\). Consequently, whether \(y\) can
receive value \(0\) can be determined from the value-\(2\) vertices of
\(P_k\) already recorded in the state. If \(y\) receives value \(2\), then
all processed arm vertices that it dominates are also known at that moment.

\paragraph*{\underline{\textbf{Pending vertices}}}

Suppose that a processed vertex \(x_{k,j}\) is assigned value \(0\). If
\(\operatorname{pre}_S(x_{k,j})=0\), then it is already dominated by a
vertex of \(S\). Otherwise, it must be dominated by a processed or future
vertex of \(Y_k\) assigned value \(2\). Until such a vertex is selected, we
call \(x_{k,j}\) \emph{pending}.

As in the convex-bipartite algorithm, it is sufficient to remember only the
earliest pending vertex on the arm. Suppose that \(x_{k,p}\) is the earliest
pending vertex. Every other pending vertex has an index at least \(p\). A future
arm vertex \(y\), when processed, has an interval neighbourhood whose
processed part is a suffix ending at its right endpoint. If \(l(y)\leq p\),
then \(y\) contains \(x_{k,p}\) and every later pending vertex, so assigning
value \(2\) to \(y\) clears the entire pending set. If \(l(y)>p\), then
\(x_{k,p}\notin N(y)\), and the earliest pending vertex remains pending even
if \(y\) dominates some later pending vertices. Hence, the minimum pending
index contains all information needed for future transitions.

\paragraph*{\underline{\textbf{State definition and sentinel values}}}

For each fixed arm \(P_k\), we distinguish state sentinels from the numerical
value \(+\infty\) used for an infeasible table entry. Let
\(a_{\infty,k}=m_k+1\) and \(p_{\infty,k}=m_k+1\). The symbol
\(a_{\infty,k}\) means that no processed vertex of \(P_k\) has been assigned
value \(2\), while \(p_{\infty,k}\) means that no processed vertex of \(P_k\)
is pending. Although the two sentinels are represented numerically by the
same integer, their different symbols make their roles explicit. The value
\(+\infty\) is reserved only for an infeasible dynamic programming entry.

The state on arm \(P_k\) is a triple \((a,b,p)\). Its coordinates have the
following meanings.

\begin{itemize}
	\item The value \(a\) is the smallest index \(i\) such that a processed
	vertex \(x_{k,i}\) has been assigned value \(2\). If no such vertex
	exists, then \(a=a_{\infty,k}\).
	
	\item The value \(b\) is the largest index \(i\) such that a processed
	vertex \(x_{k,i}\) has been assigned value \(2\). If no such vertex
	exists, then \(b=0\).
	
	\item The value \(p\) is the smallest index of a processed vertex
	\(x_{k,p}\) such that \(f(x_{k,p})=0\),
	\(\operatorname{pre}_S(x_{k,p})=1\), and \(x_{k,p}\) has not yet been
	dominated by a processed vertex of \(Y_k\) assigned value \(2\). If no
	such vertex exists, then \(p=p_{\infty,k}\).
\end{itemize}

The table entry \(D_k(a,b,p)\) stores the minimum weight of a partial
assignment on the processed vertices of \(P_k\cup Y_k\) that realizes
\((a,b,p)\). The fixed contribution \(2|S|+\alpha\) is not included in
\(D_k\). An entry has value \(+\infty\) if no feasible partial assignment
realizes the corresponding state.

At every stage of the arm computation, the table satisfies the following
conditions.

\begin{itemize}
	\item Every processed vertex of \(Y_k\) assigned value \(0\) has a
	processed neighbour in \(P_k\) assigned value \(2\).
	
	\item Every processed vertex of \(P_k\) assigned value \(0\) is already
	dominated by a vertex of \(S\), dominated by a processed vertex of
	\(Y_k\) assigned value \(2\), or represented by the pending coordinate
	\(p\).
	
	\item The coordinates \(a\) and \(b\) are respectively the smallest and
	largest indices of the processed vertices of \(P_k\) assigned value
	\(2\).
\end{itemize}

Before any vertex of the arm is processed, the only feasible state is
\((a_{\infty,k},0,p_{\infty,k})\), with value \(0\). Thus, we initialize
\(D_k(a_{\infty,k},0,p_{\infty,k})=0\) and set every other table entry to
\(+\infty\).

\paragraph*{\underline{\textbf{Why both \texorpdfstring{\(a\)} and \texorpdfstring{\(b\)} are stored}}}

The coordinate \(b\) is needed while processing the arm vertices in \(Y_k\).
When \(y\in Y_{k,i}\) is assigned value \(0\), it must have a processed
value-\(2\) neighbour in
\(N(y)=\{x_{k,l(y)},\ldots,x_{k,i}\}\). Such a neighbour exists exactly when
\(b\geq l(y)\).

The coordinate \(a\) is needed after the arm computation. Since the arm is
ordered from the center outward, \(a\) records the value-\(2\) vertex of the
arm closest to the center. A central \(Y\)-vertex \(y\) reaches the prefix
\(\{x_{k,1},\ldots,x_{k,r_k(y)}\}\) of \(P_k\). Therefore, \(y\) has a
value-\(2\) neighbour on \(P_k\) exactly when \(a\leq r_k(y)\). The exact
set of value-\(2\) vertices on the arm is not needed: the rightmost one,
recorded by \(b\), is sufficient for interval vertices, and the closest one,
recorded by \(a\), is sufficient for central vertices.

\paragraph*{\underline{\textbf{Transitions for a vertex of an arm}}}

Suppose that the current arm state is \((a,b,p)\), and that \(x_{k,i}\) is
the next vertex to be processed. We consider the three possible values of
\(f(x_{k,i})\).

\begin{enumerate}
	\item Suppose that \(f(x_{k,i})=0\). If
	\(\operatorname{pre}_S(x_{k,i})=0\), then \(x_{k,i}\) is already
	dominated by a vertex of \(S\), so the pending coordinate remains \(p\).
	If \(\operatorname{pre}_S(x_{k,i})=1\), then \(x_{k,i}\) must become
	pending, and the new pending coordinate is \(\min\{p,i\}\). This
	transition adds no weight.
	
	\item Suppose that \(f(x_{k,i})=1\). The vertex is safe by itself and
	does not affect any domination obligation. The state remains
	\((a,b,p)\), and the weight increases by \(1\).
	
	\item Suppose that \(f(x_{k,i})=2\). Since the vertices of \(P_k\) are
	processed in increasing order, \(x_{k,i}\) becomes the farthest
	processed value-\(2\) vertex from the center. The new state is
	\((\min\{a,i\},i,p)\), and the weight increases by \(2\). The pending
	coordinate does not change, because there are no edges between vertices
	of \(X\).
\end{enumerate}

Whenever more than one transition produces the same state, only the minimum
resulting weight is retained.

\paragraph*{\underline{\textbf{Transitions for an arm vertex of \texorpdfstring{\(Y_k\)}}}}

After processing \(x_{k,i}\), the vertices of \(Y_{k,i}\) are processed one
at a time. Let \(y\in Y_{k,i}\), \(N(y)=\{x_{k,l(y)},\ldots,x_{k,i}\}\). Again, we consider the three possible
values of \(f(y)\).

\begin{enumerate}
	\item Suppose that \(f(y)=0\). Since all neighbours of \(y\) have already
	been processed, this choice is feasible exactly when \(y\) already has a
	value-\(2\) neighbour on the arm. The farthest processed value-\(2\)
	vertex has index \(b\), so such a neighbour exists exactly when
	\(b\geq l(y)\). If this condition holds, the state does not change, and
	the weight does not increase. Otherwise, the choice is infeasible.
	
	\item Suppose that \(f(y)=1\). The vertex \(y\) is safe by itself, and it
	does not remove a pending obligation. Hence, the state remains
	\((a,b,p)\), and the weight increases by \(1\).
	
	\item Suppose that \(f(y)=2\). The vertex \(y\) dominates every pending
	vertex in its interval. If \(p\neq p_{\infty,k}\) and \(p\geq l(y)\),
	then the earliest pending vertex belongs to \(N(y)\). Every later pending
	vertex also lies in \(N(y)\), so the new pending coordinate becomes
	\(p_{\infty,k}\). If \(p=p_{\infty,k}\) or \(p<l(y)\), then the pending
	coordinate remains \(p\). The values \(a\) and \(b\) do not change,
	because they record only value-\(2\) vertices of \(P_k\). This
	transition increases the weight by \(2\).
\end{enumerate}

An arm vertex \(y\in Y_k\) cannot be dominated by a central vertex of \(S\),
because both belong to the bipartition class \(Y\). Therefore, the
feasibility of assigning value \(0\) to \(y\) depends only on the
value-\(2\) vertices of \(P_k\), as recorded by \(b\).

\paragraph*{\underline{\textbf{Terminal arm states}}}

After all vertices of \(P_k\cup Y_k\) have been processed, we retain only
states with \(p=p_{\infty,k}\). This condition ensures that every vertex of
\(P_k\) assigned value \(0\) is dominated either by a vertex of \(S\) or by
an arm vertex of \(Y_k\) assigned value \(2\). Every vertex of \(Y_k\)
assigned value \(0\) was checked for domination at the time it was
processed.

For each \(t\in\{1,\ldots,m_k\}\cup\{a_{\infty,k}\}\), let \(A_k(t)\) be the
minimum weight of an accepted arm assignment whose closest value-\(2\)
vertex on \(P_k\) has index \(t\). Thus,
\(A_k(t)=\min\{D_k(t,b,p_{\infty,k}):0\leq b\leq m_k\}\). The value
\(t=a_{\infty,k}\) means that no vertex of \(P_k\) receives value \(2\).

\subsection{Handling the Remaining Central Vertices}

After the three arm dynamic programs have been solved for a fixed branch
\((S,\alpha)\), it remains to assign values to the vertices of
\(Y^C\setminus S\). By the definition of the branch, these vertices are not
allowed to receive value \(2\). Hence, each such vertex receives value \(0\)
if it has a value-\(2\) neighbour in \(X\), and otherwise it must receive
value \(1\).

Let \(t_1,t_2,t_3\) be the closest value-\(2\) indices returned by the three
arm computations. A vertex \(y\in Y^C\setminus S\) is dominated by the
center \(c\) when \(\alpha=2\). It is dominated by a value-\(2\) vertex on
arm \(P_k\) exactly when \(t_k\leq r_k(y)\). This condition is false
automatically when \(t_k=a_{\infty,k}\), since then
\(a_{\infty,k}=m_k+1>r_k(y)\). Consequently, \(y\) can receive value \(0\)
exactly when \(\alpha=2\), or when \(t_k\leq r_k(y)\) for at least one
\(k\in\{1,2,3\}\).

For a fixed branch \((S,\alpha)\) and fixed indices \(t_1,t_2,t_3\), let
\(\kappa(S,\alpha,t_1,t_2,t_3)\) be the number of vertices
\(y\in Y^C\setminus S\) such that \(\alpha\neq 2\) and
\(t_k>r_k(y)\) for every \(k\in\{1,2,3\}\). These are exactly the remaining
central \(Y\)-vertices that must receive value \(1\). Once the values
\(t_1,t_2,t_3\) are fixed, the choices for the vertices of
\(Y^C\setminus S\) are independent because there are no edges inside \(Y\).

Therefore, the total weight associated with this branch and this triple is
\(2|S|+\alpha+A_1(t_1)+A_2(t_2)+A_3(t_3)+
\kappa(S,\alpha,t_1,t_2,t_3)\). The minimum weight for the branch is obtained
by taking the minimum of this expression over all finite choices of
\(t_1,t_2,t_3\).

\subsection{Algorithm}

We divide the algorithm into three parts.
Algorithm~\ref{alg:roman-triad-convex} is the main routine,
Algorithm~\ref{alg:process-triad-x} processes a vertex of an arm, and
Algorithm~\ref{alg:process-triad-y} processes an arm vertex belonging to
\(Y\).

For each arm \(P_k\), the values \(a_{\infty,k}\) and \(p_{\infty,k}\) are
special state values, whereas \(+\infty\) denotes an infeasible table entry.
The value \(a_{\infty,k}\) means that no processed vertex of \(P_k\) has been
assigned value \(2\), while \(p_{\infty,k}\) means that no processed vertex
of \(P_k\) is currently pending. Both values are represented numerically by
\(m_k+1\).

For a fixed branch \((S,\alpha)\), the main routine runs an independent
dynamic program on each arm \(P_k\). The vertices of the arm instance
\(P_k\cup Y_k\) are processed in the order
\(x_{k,1},Y_{k,1},x_{k,2},Y_{k,2},\ldots,
x_{k,m_k},Y_{k,m_k}\). The subroutine \textsc{Process-Triad-X} applies the
three transitions corresponding to assigning value \(0\), \(1\), or \(2\)
to a vertex \(x_{k,i}\), whereas \textsc{Process-Triad-Y} applies the
corresponding transitions to a vertex \(y\in Y_k\).

A state is called \emph{finite} if its table entry is smaller than
\(+\infty\); equivalently, it is realized by at least one feasible partial
assignment. Whenever a vertex is processed, the corresponding subroutine
initializes every entry of a temporary table \(D_k'\) to \(+\infty\). The
instruction ``update \(D_k'(a,b,p)\) with \(w\)'' means that
\(D_k'(a,b,p)\) is replaced by
\(\min\{D_k'(a,b,p),w\}\). Thus, if several transitions produce the same
state, only the minimum resulting weight is retained. After all transitions
for the current vertex have been considered, \(D_k'\) is returned to the
main routine and replaces the current table \(D_k\).

After arm \(P_k\) has been processed completely, the value \(A_k(t)\) is
finite when \(A_k(t)<+\infty\). Such a value represents the minimum weight
of an accepted assignment on \(P_k\cup Y_k\) whose closest value-\(2\)
vertex on the arm has index \(t\). The special choice
\(t=a_{\infty,k}\) means that no vertex of \(P_k\) receives value \(2\).
After the three arm tables have been computed, the main routine examines
every triple \((t_1,t_2,t_3)\) for which
\(A_1(t_1)\), \(A_2(t_2)\), and \(A_3(t_3)\) are finite. It then adds the
fixed contribution \(2|S|+\alpha\) and the value \(\kappa\), which is the
minimum additional contribution required for the vertices of
\(Y^C\setminus S\).

\begin{algorithm}[H]
	\DontPrintSemicolon
	\caption{\textsc{Roman-Triad-Convex-Bipartite}}
	\label{alg:roman-triad-convex}
	\KwIn{A triad-convex bipartite graph \(G=(X\cup Y,E)\) with a
		triad-convex representation on \(X\).}
	\KwOut{The Roman domination number \(\gamma_R(G)\).}
	
	Let \(c\) be the center, and let \(P_1,P_2,P_3\) be the three arms\;
	Partition \(Y\) into \(Y^C,Y_1,Y_2,Y_3\)\;
	
	\For{\(k=1\) \KwTo \(3\)}{
		For every \(i\in\{1,\ldots,m_k\}\), set
		\(Y_{k,i}=\{y\in Y_k:r(y)=i\}\)\;
		Set \(a_{\infty,k}=m_k+1\) and
		\(p_{\infty,k}=m_k+1\)\;
	}
	
	Set \(\gamma_R=+\infty\)\;
	
	\ForEach{\(S\subseteq Y^C\) with \(|S|\leq 3\)}{
		\ForEach{\(\alpha\in\{0,1,2\}\)}{
			\If{\(\alpha\neq 0\) or \(S\neq\emptyset\)}{
				\For{\(k=1\) \KwTo \(3\)}{
					\ForEach{\(x_{k,i}\in P_k\)}{
						\eIf{\(N(x_{k,i})\cap S\neq\emptyset\)}{
							Set \(\operatorname{pre}_S(x_{k,i})=0\)\;
						}{
							Set \(\operatorname{pre}_S(x_{k,i})=1\)\;
						}
					}
					
					Set all entries of \(D_k\) to \(+\infty\)\;
					Set
					\(D_k(a_{\infty,k},0,p_{\infty,k})=0\)\;
					
					\For{\(i=1\) \KwTo \(m_k\)}{
						Set
						\(D_k=\textsc{Process-Triad-X}
						(D_k,k,i,\operatorname{pre}_S(x_{k,i}))\)\;
						
						\ForEach{\(y\in Y_{k,i}\), in an arbitrary fixed order}{
							Set
							\(D_k=\textsc{Process-Triad-Y}(D_k,k,y)\)\;
						}
					}
					
					\ForEach{\(t\in
						\{1,\ldots,m_k\}\cup\{a_{\infty,k}\}\)}{
						Set
						\(A_k(t)=
						\min\{D_k(t,b,p_{\infty,k}):
						0\leq b\leq m_k\}\)\;
					}
				}
				
				\ForEach{\((t_1,t_2,t_3)\) such that
					\(A_k(t_k)<+\infty\) for every
					\(k\in\{1,2,3\}\)}{
					Set
					\(\kappa=
					|\{y\in Y^C\setminus S:
					\alpha\neq 2\text{ and }
					t_k>r_k(y)\text{ for every }
					k\in\{1,2,3\}\}|\)\;
					
					Set
					\(\gamma_R=\min\{\gamma_R,
					2|S|+\alpha+A_1(t_1)+A_2(t_2)
					+A_3(t_3)+\kappa\}\)\;
				}
			}
		}
	}
	
	\Return{\(\gamma_R\)}\;
\end{algorithm}

\begin{algorithm}[H]
	\DontPrintSemicolon
	\caption{\textsc{Process-Triad-X}}
	\label{alg:process-triad-x}
	\KwIn{The current arm table \(D_k\), the arm index \(k\), an index
		\(i\), and \(\operatorname{pre}_S(x_{k,i})\).}
	\KwOut{The updated table after processing \(x_{k,i}\).}
	
	Set all entries of \(D_k'\) to \(+\infty\)\;
	
	\ForEach{finite state \((a,b,p)\) of \(D_k\)}{
		\eIf{\(\operatorname{pre}_S(x_{k,i})=0\)}{
			Set \(p_0=p\)\;
		}{
			Set \(p_0=\min\{p,i\}\)\;
		}
		
		Update \(D_k'(a,b,p_0)\) with \(D_k(a,b,p)\)
		\tcp*{\(f(x_{k,i})=0\)}
		
		Update \(D_k'(a,b,p)\) with \(D_k(a,b,p)+1\)
		\tcp*{\(f(x_{k,i})=1\)}
		
		Set \(a_2=\min\{a,i\}\)\;
		Update \(D_k'(a_2,i,p)\) with \(D_k(a,b,p)+2\)
		\tcp*{\(f(x_{k,i})=2\)}
	}
	
	\Return{\(D_k'\)}\;
\end{algorithm}

\begin{algorithm}[H]
	\DontPrintSemicolon
	\caption{\textsc{Process-Triad-Y}}
	\label{alg:process-triad-y}
	\KwIn{The current arm table \(D_k\), the arm index \(k\), and a vertex
		\(y\in Y_k\).}
	\KwOut{The updated table after processing \(y\).}
	
	Set \(l=l(y)\)\;
	Set all entries of \(D_k'\) to \(+\infty\)\;
	
	\ForEach{finite state \((a,b,p)\) of \(D_k\)}{
		\If{\(b\geq l\)}{
			Update \(D_k'(a,b,p)\) with \(D_k(a,b,p)\)
			\tcp*{\(f(y)=0\)}
		}
		
		Update \(D_k'(a,b,p)\) with \(D_k(a,b,p)+1\)
		\tcp*{\(f(y)=1\)}
		
		\eIf{\(p\neq p_{\infty,k}\) and \(p\geq l\)}{
			Set \(p_2=p_{\infty,k}\)\;
		}{
			Set \(p_2=p\)\;
		}
		
		Update \(D_k'(a,b,p_2)\) with \(D_k(a,b,p)+2\)
		\tcp*{\(f(y)=2\)}
	}
	
	\Return{\(D_k'\)}\;
\end{algorithm}

\subsection{Correctness Proof}

We prove that Algorithm~\ref{alg:roman-triad-convex} computes
\(\gamma_R(G)\).

Fix a branch consisting of a set \(S\subseteq Y^C\) with \(|S|\leq 3\) and
a value \(\alpha\in\{0,1,2\}\) for the center \(c\). Every vertex of \(S\)
is forced to receive value \(2\), every vertex of \(Y^C\setminus S\) is
forbidden from receiving value \(2\), and \(c\) receives value \(\alpha\).
If \(\alpha=0\) and \(S=\emptyset\), then \(c\) has no value-\(2\)
neighbour, so the algorithm correctly discards this branch.

Consider an arm \(P_k\). The function \(\operatorname{pre}_S\) records
exactly which vertices of \(P_k\) are already dominated by the forced
value-\(2\) vertices of \(S\). Since
\(\operatorname{pre}_S(x_{k,i})=0\) precisely for these vertices, assigning
value \(0\) to one of them creates no pending obligation. Every other
value-\(0\) vertex of \(P_k\) is kept pending until it is dominated by a
value-\(2\) vertex of \(Y_k\).

The state \((a,b,p)\) stores all information required for the remainder of
the arm computation. The value \(p\) records the earliest pending vertex.
This is sufficient because every future vertex of \(Y_k\) sees the processed
part of \(P_k\) as a suffix. Therefore, if a future value-\(2\) vertex of
\(Y_k\) reaches the earliest pending vertex, then it reaches every later
pending vertex; otherwise, the earliest pending vertex remains pending. The
coordinate \(b\) determines whether a processed vertex of \(Y_k\) assigned
value \(0\) has a value-\(2\) neighbour on the arm, while \(a\) records the
closest value-\(2\) vertex needed to handle the central vertices after the
arm computation.

An induction on the processing order proves the following invariant. After
each processing step on arm \(P_k\), if a feasible partial assignment
realizes state \((a,b,p)\), then \(D_k(a,b,p)\) is the minimum weight among
all such assignments. If no feasible partial assignment realizes the state,
then \(D_k(a,b,p)=+\infty\).

The invariant holds initially because
\(D_k(a_{\infty,k},0,p_{\infty,k})=0\) represents the empty arm assignment,
while every other state is infeasible. When \(x_{k,i}\) is processed, the
three transitions consider all possible values of \(f(x_{k,i})\) and
correctly determine whether the vertex becomes pending. When a vertex
\(y\in Y_{k,i}\) is processed, the transition for \(f(y)=0\) is allowed
exactly when \(b\geq l(y)\), and the transition for \(f(y)=2\) clears the
pending set exactly when the earliest pending index lies in \(N(y)\).
Therefore, every feasible extension is considered, and every generated
extension is feasible. Retaining the minimum value for each resulting state
preserves the invariant.

At the end of the arm computation, only states with
\(p=p_{\infty,k}\) are accepted. Thus, every value-\(0\) vertex of \(P_k\)
is dominated either by a vertex of \(S\) or by a value-\(2\) vertex of
\(Y_k\), and every value-\(0\) vertex of \(Y_k\) was verified when it was
processed. Consequently, \(A_k(t)\) is exactly the minimum arm weight among
accepted assignments whose closest value-\(2\) vertex has index \(t\).

After the three arms have been solved, the only unassigned vertices are
those in \(Y^C\setminus S\). Such a vertex \(y\) can receive value \(0\)
exactly when \(\alpha=2\), or when \(t_k\leq r_k(y)\) for at least one arm
\(P_k\). Otherwise, \(y\) must receive value \(1\). Hence,
\(\kappa(S,\alpha,t_1,t_2,t_3)\) is exactly the minimum additional cost for
the remaining central vertices.

It remains to show that the branching is sufficient. By
Lemma~\ref{lem:triad-three-central}, there exists an optimal Roman
dominating function in which at most three central \(Y\)-vertices receive
value \(2\). Let \(S\) be exactly this set, and let \(\alpha\) be the value
assigned to \(c\). The algorithm considers this branch. For this branch,
each arm table finds the best feasible arm assignment for every possible
closest value-\(2\) index, and the final combination chooses the best triple
and the best values for the remaining central vertices. Therefore, the
algorithm finds a solution of weight at most \(\gamma_R(G)\).

Conversely, every solution constructed by the algorithm is a valid Roman
dominating function. The vertices of \(S\) receive value \(2\), the center
receives value \(\alpha\), every arm is completed with no pending vertex,
and each remaining central \(Y\)-vertex receives value \(0\) only when it
has a value-\(2\) neighbour in \(X\); otherwise, it receives value \(1\).
Thus, the algorithm never constructs an invalid solution and never
underestimates \(\gamma_R(G)\). Therefore, the returned value is exactly
\(\gamma_R(G)\).

\subsection{Running Time Analysis}

\begin{theorem}
	\label{thm:triad-convex-rdp}
	Given a triad-convex representation of a triad-convex bipartite graph \(G\)
	on \(n\) vertices, Algorithm~\ref{alg:roman-triad-convex} computes
	\(\gamma_R(G)\) in \(O(n^7)\) time.
\end{theorem}

\begin{proof}
	There are at most
	\(1+|Y^C|+\binom{|Y^C|}{2}+\binom{|Y^C|}{3}=O(n^3)\) choices for the set
	\(S\subseteq Y^C\) with \(|S|\leq 3\), and there are three choices for
	\(\alpha\).
	
	Fix a branch \((S,\alpha)\). On arm \(P_k\), the coordinate \(a\) has
	\(m_k+1\) possible values, namely
	\(\{1,\ldots,m_k\}\cup\{a_{\infty,k}\}\); the coordinate \(b\) has
	\(m_k+1\) possible values, namely \(\{0,1,\ldots,m_k\}\); and the coordinate
	\(p\) has \(m_k+1\) possible values, namely
	\(\{1,\ldots,m_k\}\cup\{p_{\infty,k}\}\). Hence, the number of states on
	arm \(P_k\) is \(O(m_k^3)\).
	
	Each vertex of \(P_k\cup Y_k\) is processed once, and each finite state
	produces at most three transitions. Every transition takes constant time.
	Thus, the time required for arm \(P_k\) is
	\(O((m_k+|Y_k|)m_k^3)\). Summing over the three arms gives \(O(n^4)\) time
	in the worst case.
	
	There are at most
	\((m_1+1)(m_2+1)(m_3+1)=O(n^3)\) triples
	\((t_1,t_2,t_3)\). For each triple, the value \(\kappa\) can be computed by
	scanning \(Y^C\), which takes \(O(n)\) time. Therefore, the combination step
	takes \(O(n^4)\) time for one fixed branch. Computing all
	\(\operatorname{pre}_S\) values and grouping the arm vertices by their right
	endpoints take no more than this bound.
	
	Hence, each branch takes \(O(n^4)\) time. Since there are \(O(n^3)\) choices
	for \(S\) and only three choices for \(\alpha\), the total running time is
	\(O(n^7)\).
\end{proof}

\begin{remark}
	The algorithm assumes that a triad-convex representation is given. It uses
	the center \(c\), the ordered arms \(P_1,P_2,P_3\), and the interval
	endpoints of every arm vertex. Recognizing triad-convex bipartite graphs and
	constructing a representation are separate preprocessing questions and are
	not part of the dynamic program described here.
\end{remark}

\subsection{Example}

We illustrate Algorithm~\textsc{Roman-Triad-Convex-Bipartite} on a small
triad-convex bipartite graph. Let the triad have center \(c\) and arms
\(P_1,P_2,P_3\), where \(P_1=\{x_{1,1},x_{1,2}\}\),
\(P_2=\{x_{2,1}\}\), and \(P_3=\{x_{3,1}\}\). Let
\(Y=\{y_1,y_2,y_3\}\). The vertices \(y_1\) and \(y_2\) are central
vertices, while \(y_3\) is an arm vertex on \(P_1\). Define
\(N(y_1)=\{c,x_{1,1},x_{2,1},x_{3,1}\}\),
\(N(y_2)=\{c,x_{1,1},x_{1,2}\}\), and
\(N(y_3)=\{x_{1,2}\}\). Thus, \(r_1(y_1)=r_2(y_1)=r_3(y_1)=1\), while
\(r_1(y_2)=2\) and \(r_2(y_2)=r_3(y_2)=0\). Hence,
\(Y^C=\{y_1,y_2\}\), \(Y_1=\{y_3\}\), and there are no arm vertices on
\(P_2\) or \(P_3\). The graph and the Roman dominating function obtained by
the algorithm are shown in Figure~\ref{fig:roman-triad-convex-example}.

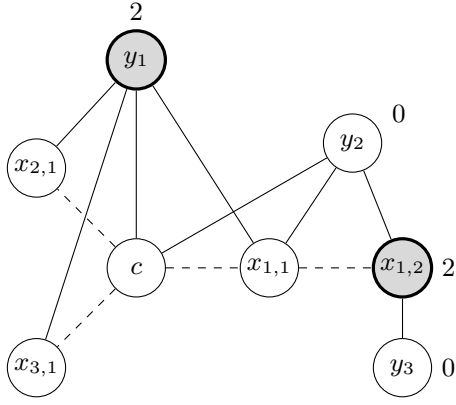
\begin{figure}[h]
	\centering
	\begin{tikzpicture}[
		scale=1.10,
		transform shape,
		xnode/.style={circle, draw, minimum size=7mm, inner sep=0pt},
		ynode/.style={circle, draw, minimum size=7mm, inner sep=0pt},
		opt2/.style={
			circle,
			draw,
			fill=gray!30,
			very thick,
			minimum size=7mm,
			inner sep=0pt
		},
		every node/.style={font=\small}
		]
		\node[xnode] (c) at (0,0) {\(c\)};
		\node[xnode] (x11) at (1.6,0) {\(x_{1,1}\)};
		\node[opt2]  (x12) at (3.2,0) {\(x_{1,2}\)};
		\node[xnode] (x21) at (-1.2,1.2) {\(x_{2,1}\)};
		\node[xnode] (x31) at (-1.2,-1.2) {\(x_{3,1}\)};
		
		\draw[dashed] (c)--(x11)--(x12);
		\draw[dashed] (c)--(x21);
		\draw[dashed] (c)--(x31);
		
		\node[opt2]  (y1) at (0,2.5) {\(y_1\)};
		\node[ynode] (y2) at (2.6,1.5) {\(y_2\)};
		\node[ynode] (y3) at (3.2,-1.2) {\(y_3\)};
		
		\draw (y1)--(c);
		\draw (y1)--(x11);
		\draw (y1)--(x21);
		\draw (y1)--(x31);
		
		\draw (y2)--(c);
		\draw (y2)--(x11);
		\draw (y2)--(x12);
		
		\draw (y3)--(x12);
		
		\node[above] at (0,2.85) {\(2\)};
		\node[right] at (3.55,0) {\(2\)};
		\node[right] at (2.95,1.85) {\(0\)};
		\node[right] at (3.55,-1.2) {\(0\)};
	\end{tikzpicture}
	\caption{A triad-convex bipartite graph used in the example. The shaded
		vertices receive value \(2\), namely \(f(y_1)=2\) and
		\(f(x_{1,2})=2\). All other vertices receive value \(0\).}
	\label{fig:roman-triad-convex-example}
\end{figure}

We now trace the branch illustrated in
Figure~\ref{fig:roman-triad-convex-example}. Take \(S=\{y_1\}\) and
\(\alpha=f(c)=0\). This branch is feasible because \(c\) is dominated by
\(y_1\), and the fixed contribution is \(2|S|+\alpha=2\). The vertex \(y_1\)
pre-dominates \(x_{1,1}\), \(x_{2,1}\), and \(x_{3,1}\), but it does not
pre-dominate \(x_{1,2}\). Therefore,
\(\operatorname{pre}_S(x_{1,1})=
\operatorname{pre}_S(x_{2,1})=
\operatorname{pre}_S(x_{3,1})=0\), while
\(\operatorname{pre}_S(x_{1,2})=1\).

On arm \(P_1\), we have \(a_{\infty,1}=p_{\infty,1}=3\), and the processing
order is \(x_{1,1},x_{1,2},y_3\). The following table traces one optimal arm
branch.

\begin{center}
	\large
	\begin{tabular}{c|c|c}
		Step & Chosen value & Resulting state and value \\
		\hline
		Initial
		& --
		& \(D_1(a_{\infty,1},0,p_{\infty,1})=0\) \\
		
		After \(x_{1,1}\)
		& \(f(x_{1,1})=0\), \(\operatorname{pre}_S(x_{1,1})=0\)
		& \(D_1(a_{\infty,1},0,p_{\infty,1})=0\) \\
		
		After \(x_{1,2}\)
		& \(f(x_{1,2})=2\)
		& \(D_1(2,2,p_{\infty,1})=2\) \\
		
		After \(y_3\)
		& \(f(y_3)=0\)
		& \(D_1(2,2,p_{\infty,1})=2\)
	\end{tabular}
\end{center}

The first step creates no pending vertex because \(x_{1,1}\) is already
dominated by \(y_1\). The vertex \(x_{1,2}\) receives value \(2\), so it is
both the closest and the farthest value-\(2\) vertex on \(P_1\). Finally,
\(y_3\) may receive value \(0\), because it is adjacent to
\(x_{1,2}\). Hence, \(A_1(2)=2\).

On arm \(P_2\), the only vertex \(x_{2,1}\) is already dominated by \(y_1\).
We assign it value \(0\), so no pending vertex is created and
\(A_2(a_{\infty,2})=0\). Similarly,
\(A_3(a_{\infty,3})=0\).

It remains to handle \(y_2\in Y^C\setminus S\). Since \(\alpha=0\), the
center does not dominate \(y_2\). However, the closest value-\(2\) vertex on
\(P_1\) has index \(t_1=2\), and \(t_1=2\leq r_1(y_2)=2\). Hence, \(y_2\)
is dominated by \(x_{1,2}\) and may receive value \(0\). Therefore,
\(\kappa(S,\alpha,2,a_{\infty,2},a_{\infty,3})=0\).

The total weight of this branch is
\(2|S|+\alpha+A_1(2)+A_2(a_{\infty,2})+
A_3(a_{\infty,3})+
\kappa(S,\alpha,2,a_{\infty,2},a_{\infty,3})
=2+0+2+0+0+0=4\). Thus, the algorithm finds the Roman dominating function
shown in Figure~\ref{fig:roman-triad-convex-example}, namely
\(f(y_1)=2\), \(f(x_{1,2})=2\), and every other vertex receives value \(0\).

This value is optimal. Suppose, for a contradiction, that there exists a
Roman dominating function of weight at most \(3\). Such a function must have
exactly one vertex assigned value \(2\), together with at most one vertex
assigned value \(1\), because a function with no value-\(2\) vertex must
assign value \(1\) to all eight vertices.

If the unique value-\(2\) vertex belongs to \(X\), then every other vertex of
\(X\) must receive a positive value, because there are no edges inside \(X\).
There are four such vertices, so the weight is at least \(2+4=6\).

Now suppose that the unique value-\(2\) vertex belongs to \(Y\). If it is
\(y_1\), then \(x_{1,2}\) is not dominated, while both \(y_2\) and \(y_3\)
also have no value-\(2\) neighbour in \(X\). Thus, at least three additional
vertices must receive positive values. If it is \(y_2\), then both
\(x_{2,1}\) and \(x_{3,1}\) are not dominated, and the vertices \(y_1\) and
\(y_3\) also have no value-\(2\) neighbour in \(X\). If it is \(y_3\), then
the vertices \(c,x_{1,1},x_{2,1},x_{3,1}\) are not dominated, and both
\(y_1\) and \(y_2\) have no value-\(2\) neighbour in \(X\). In every case,
more than one additional vertex must receive a positive value, contradicting
the assumed weight bound of \(3\).

Therefore, no Roman dominating function of weight at most \(3\) exists, and
\(\gamma_R(G)=4\).

\section{Roman Domination on \(P_4\)-Tidy Graphs}\label{sec:rd-p4-tidy}

In this section, we show that the Roman Domination Problem is linear-time solvable on \(P_4\)-tidy graphs by using the structural decomposition of Giakoumakis et al.~\cite{giakoumakis1997p_4}. The class of cographs is the subclass of \(P_4\)-tidy graphs in which no induced \(P_4\) occurs. Roman Domination is already known to be solvable in \(O(n+m)\) time on cographs by the algorithm of Liedloff et al.~\cite{liedloff2008efficient}. Thus, the result of this section extends the known tractability from cographs to the larger graph classes, namely \(P_4\)-tidy graphs. In addition to disjoint-union and join nodes, the decomposition contains spider, quasi-spider and exceptional nodes corresponding to \(P_5\), \(\overline{P_5}\), and \(C_5\).

\paragraph{\underline{\textbf{Bounded clique-width viewpoint}}}

There are two general bounded-clique-width arguments that already imply tractability. First, Roman Domination is a LinEMSOL optimization problem. Indeed, a Roman dominating function is represented by a partition \(V(G)=V_0\cup V_1\cup V_2\), where every vertex in \(V_i\) receives value \(i\), the feasibility condition requires every vertex of \(V_0\) to have a neighbor in \(V_2\), and the objective is to minimize \(|V_1|+2|V_2|\). Consequently, the LinEMSOL meta-theorem of Courcelle et al.~\cite{courcelle2000linear} yields a linear-time algorithm on every graph class of bounded clique-width, provided that a clique-width expression of bounded width is supplied.

Roman Domination also fits directly into the framework of locally checkable problems parameterized by clique-width developed by Baghirova et al.~\cite{baghirova2022locally}. Use the three colors \(\{0,1,2\}\), and let \(n_j\) denote the number of neighbours of a vertex that receive color \(j\). The local check accepts a tuple \((v,a,n_0,n_1,n_2)\) exactly when \(a\in\{1,2\}\) or \(n_2\geq 1\). Thus, a coloring accepted at every vertex is precisely a Roman dominating function. The check is color-counting and \(1\)-stable, since it only distinguishes between \(n_2=0\) and \(n_2\geq 1\). Assigning local weights \(0\), \(1\), and \(2\) to the three colors makes the total coloring weight equal to the Roman weight. Hence, the algorithm of Baghirova et al.~\cite{baghirova2022locally} yields an FPT algorithm parameterized by clique-width, provided that a clique-width expression is given; in particular, it is linear-time solvable on every class of bounded clique-width. The structural decomposition below also shows that \(P_4\)-tidy graphs have bounded clique-width: disjoint union and join preserve bounded clique-width, the exceptional graphs have constant size, and a spider or quasi-spider is obtained from its recursively decomposed head by attaching a constant-label body and legs. Nevertheless, our purpose is to give a direct, exact algorithm with explicit recurrence rules and a linear running time on the decomposition tree.

\paragraph{\underline{\textbf{The Giakoumakis et al. decomposition tree}}}

A graph is called \(P_4\)-tidy if, for every induced \(P_4\), there is at most one vertex outside that \(P_4\) that forms more than one induced \(P_4\) together with it. The structural theorem of Giakoumakis et al.~\cite{giakoumakis1997p_4} states that every induced subgraph \(H\) of a \(P_4\)-tidy graph satisfies one of the following: \(H\cong K_1\); \(H\) is disconnected; \(\overline{H}\) is disconnected; \(H\) is a spider or a quasi-spider whose head induces a \(P_4\)-tidy graph; or \(H\) is isomorphic to one of \(P_5\), \(\overline{P_5}\), and \(C_5\).

This characterization gives a recursive decomposition tree \(T_G\). Each node \(h\) represents an induced subgraph \(G_h\). If \(G_h\cong K_1\), or \(G_h\) is isomorphic to one of \(P_5\), \(\overline{P_5}\), and \(C_5\), then \(h\) is a basic node. If \(G_h\) is disconnected, then \(h\) is a disjoint-union node whose children represent the connected components of \(G_h\). If \(\overline{G_h}\) is disconnected, then \(h\) is a join node whose children represent the connected components of \(\overline{G_h}\). Finally, if \(G_h\) is a spider or quasi-spider with partition \(S\cup C\cup R\), then \(h\) records whether the spider is thin or thick, whether one vertex of \(S\cup C\) is replaced by a pair of twins, and the partition \(S\cup C\cup R\). When \(R\neq\emptyset\), the node has a distinguished child representing the recursively decomposed head \(G_h[R]\). Thus, the exceptional graphs and \(K_1\) are terminal pieces, whereas the head of a spider or quasi-spider is decomposed recursively.

A spider is a graph whose vertex set can be partitioned into \(S\cup C\cup R\), where \(S=\{s_1,\ldots,s_k\}\) is an independent set, \(C=\{c_1,\ldots,c_k\}\) is a clique, \(k\geq 2\), every vertex of \(R\) is adjacent to every vertex of \(C\) and to no vertex of \(S\), and the adjacencies between \(S\) and \(C\) are of one of two types. In a thin spider, \(s_i\) is adjacent exactly to \(c_i\) among the vertices of \(C\). In a thick spider, \(s_i\) is adjacent exactly to the vertices of \(C\setminus\{c_i\}\). The set \(R\) is called the head. A quasi-spider is obtained from a spider by replacing exactly one vertex of \(S\cup C\) by two twins, either adjacent true twins or nonadjacent false twins, with the same neighbours outside the replacement pair as the original vertex.

For the dynamic program, we store two values for each subgraph. The first value is the Roman domination number \(\gamma_R(G)\). The second value is an auxiliary value, denoted by \(\gamma_R^+(G)\), which is the minimum weight of a Roman dominating function of \(G\) in which at least one vertex receives value \(2\). For \(K_1\), we have \(\gamma_R(K_1)=1\) and \(\gamma_R^+(K_1)=2\). For the other three exceptional graphs, the required values are \(\gamma_R(P_5)=\gamma_R^+(P_5)=4\), \(\gamma_R(\overline{P_5})=\gamma_R^+(\overline{P_5})=3\), and \(\gamma_R(C_5)=\gamma_R^+(C_5)=4\). These values are easy to verify directly. For \(P_5\), two vertices of value \(2\) at distance three give weight \(4\), while weight at most \(3\) cannot dominate both ends of the path. For \(C_5\), two vertices of value \(2\) give weight \(4\), while one vertex of value \(2\) leaves two non-neighbours that cannot both be handled with total weight at most \(3\). For \(\overline{P_5}\), one vertex of value \(2\) together with one vertex of value \(1\) gives weight \(3\), and weight \(2\) is impossible because \(\overline{P_5}\) has no universal vertex.

\subsection{Recurrence Rules}

We first describe the recurrence rules used by the algorithm. Suppose that \(G\) is disconnected, with connected components \(G_1,\ldots,G_t\). Since there are no edges between distinct components, Roman domination is additive over connected components. Hence, \(\gamma_R(G)=\sum_{i=1}^{t}\gamma_R(G_i)\). For the auxiliary value \(\gamma_R^+(G)\), at least one component must contain a vertex assigned value \(2\). Thus, if \(A=\sum_{i=1}^{t}\gamma_R(G_i)\), then \(\gamma_R^+(G)=A+\min_{1\leq i\leq t}\{\gamma_R^+(G_i)-\gamma_R(G_i)\}\).

Next, suppose that \(\overline{G}\) is disconnected. Then \(G\) is the join \(G_1\vee G_2\vee\cdots\vee G_t\), where \(G_1,\ldots,G_t\) are the connected components of \(\overline{G}\), and \(t\geq 2\). If no vertex receives value \(2\), then every vertex must receive value \(1\), giving weight \(|V(G)|\). If all vertices of value \(2\) lie in a single join-part \(G_i\), then the restriction to \(G_i\) must be a Roman dominating function of \(G_i\) with at least one vertex of value \(2\), while all vertices outside \(G_i\) may receive value \(0\), because they are adjacent to every vertex of \(G_i\). This gives cost \(\gamma_R^+(G_i)\). Finally, if at least two join-parts contain a vertex of value \(2\), then choosing one vertex of value \(2\) in each of two distinct join-parts dominates the whole graph, giving cost \(4\), and no solution of this type can have a smaller cost. Therefore \(\gamma_R^+(G)=\min\{4,\min_{1\leq i\leq t}\gamma_R^+(G_i)\}\), and \(\gamma_R(G)=\min\{|V(G)|,\gamma_R^+(G)\}\).

It remains to handle spider and quasi-spider nodes. Let \(G\) be a spider or quasi-spider with underlying spider partition \(S\cup C\cup R\), where \(S=\{s_1,\ldots,s_k\}\) and \(C=\{c_1,\ldots,c_k\}\). In a quasi-spider, one vertex of \(S\cup C\) is replaced by two twins, either adjacent or nonadjacent, and the value \(k\) refers to the size of the underlying spider before this replacement.

If the underlying spider is thin, then \(\gamma_R(G)=\gamma_R^+(G)=k+1\). For the upper bound, choose an index \(q\) so that the clique-side vertex corresponding to \(c_q\) is not replaced whenever this matters; this is possible because \(k\geq 2\). Assign value \(2\) to that clique-side vertex, assign value \(0\) to the vertex or vertices corresponding to \(s_q\), assign value \(1\) to every \(S\)-side vertex corresponding to \(s_i\) with \(i\neq q\), and assign value \(0\) to all remaining vertices. Then the vertex or vertices corresponding to \(s_q\) are dominated by the chosen value-\(2\) vertex, all clique-side vertices are dominated because the clique side is complete except possibly inside the replacement pair, and all vertices of \(R\) are dominated because \(R\) is complete to the clique side. The total weight is \(k+1\). For the lower bound, for every index \(i\), the block corresponding to \(\{s_i,c_i\}\), with the natural interpretation when one of these vertices is replaced by two twins, must contribute at least one unit of weight; otherwise, the \(S\)-side vertex or vertices corresponding to \(s_i\) are not dominated. Hence, the total weight is at least \(k\). If the total weight were exactly \(k\), then each block would contribute exactly one unit, and no vertex would receive value \(2\). Then the clique side, and also \(R\) when \(R\neq\emptyset\), could not be dominated. Thus, one additional unit is necessary, and every Roman dominating function has weight at least \(k+1\). Therefore \(\gamma_R(G)=\gamma_R^+(G)=k+1\).

If the underlying spider is thick, then \(\gamma_R(G)=\gamma_R^+(G)=3\). For the upper bound, choose a clique-side vertex corresponding to \(c_q\) so that the corresponding \(S\)-side vertex is not the replaced vertex; this is possible because \(k\geq 2\). Assign value \(2\) to this clique-side vertex, assign value \(1\) to the unique \(S\)-side vertex not adjacent to it, and assign value \(0\) to all remaining vertices. Then every other \(S\)-side vertex is adjacent to the chosen value-\(2\) vertex, every other clique-side vertex is adjacent to it, and every vertex of \(R\) is adjacent to it. Hence, the weight is \(3\). For the lower bound, a Roman dominating function of weight at most \(2\) has either no vertex of value \(2\), or exactly one vertex of value \(2\). If no vertex receives value \(2\), then no vertex can receive value \(0\), which is impossible because the graph has at least four vertices. If there is exactly one vertex of value \(2\), then that vertex cannot dominate the whole graph. If it lies on the clique side, then its private non-neighbour on the \(S\)-side is not dominated. If it lies on the \(S\)-side, then another \(S\)-side vertex is not dominated. If it lies in \(R\), then no \(S\)-side vertex is dominated. Hence, weight at most \(2\) is impossible, and the optimum value is \(3\).

\subsection{Algorithm}

We now give the algorithm formally. It assumes that the Giakoumakis-Roussel-Thuillier \(P_4\)-tidy decomposition tree described above is given. If only the graph is given, we first construct this decomposition tree using the recognition and decomposition algorithm of Giakoumakis et al.~\cite{giakoumakis1997p_4}, and then apply the bottom-up recurrence.

\begin{algorithm}[H]
	\DontPrintSemicolon
	\caption{\textsc{Roman-\(P_4\)-Tidy}}\label{alg:roman-p4-tidy}
	\KwIn{A \(P_4\)-tidy graph \(G\), together with its Giakoumakis-Roussel-Thuillier \(P_4\)-tidy decomposition tree \(T_G\).}
	\KwOut{The Roman domination number \(\gamma_R(G)\).}
	
	Process the nodes of \(T_G\) in postorder\;
	
	\ForEach{node \(h\) of \(T_G\) in postorder}{
		Let \(G_h\) be the subgraph represented by \(h\), and set \(n_h=|V(G_h)|\)\;
		
		\uIf{\(G_h\cong K_1\)}{
			Set \((\rho(h),\rho^+(h))=(1,2)\)\;
		}
		\uElseIf{\(G_h\cong P_5\) or \(G_h\cong C_5\)}{
			Set \((\rho(h),\rho^+(h))=(4,4)\)\;
		}
		\uElseIf{\(G_h\cong \overline{P_5}\)}{
			Set \((\rho(h),\rho^+(h))=(3,3)\)\;
		}
		\uElseIf{\(h\) is a disconnected node with children \(h_1,\ldots,h_t\)}{
			Set \(A=\sum_{i=1}^{t}\rho(h_i)\) and \(\delta=\min_{1\leq i\leq t}\{\rho^+(h_i)-\rho(h_i)\}\)\;
			Set \((\rho(h),\rho^+(h))=(A,A+\delta)\)\;
		}
		\uElseIf{\(h\) is a complement-disconnected node with children \(h_1,\ldots,h_t\)}{
			Set \(B=\min_{1\leq i\leq t}\rho^+(h_i)\)\;
			Set \(\rho^+(h)=\min\{4,B\}\) and \(\rho(h)=\min\{n_h,\rho^+(h)\}\)\;
		}
		\uElseIf{\(h\) is a thin spider or thin quasi-spider node with underlying spider partition \(S\cup C\cup R\)}{
			Set \(k=|S|=|C|\) in the underlying spider\;
			Set \((\rho(h),\rho^+(h))=(k+1,k+1)\)\;
		}
		\ElseIf{\(h\) is a thick spider or thick quasi-spider node}{
			Set \((\rho(h),\rho^+(h))=(3,3)\)\;
		}
	}
	
	Let \(r\) be the root of \(T_G\)\;
	\Return{\(\rho(r)\)}\;
\end{algorithm}

\subsection{Correctness Proof}

\begin{theorem}\label{thm:p4-tidy-roman-correct}
	Algorithm~\ref{alg:roman-p4-tidy} computes the Roman domination number of a \(P_4\)-tidy graph \(G\).
\end{theorem}

\begin{proof}
	We prove the correctness by induction on the Giakoumakis–Roussel–Thuillier decomposition tree. The base cases are \(K_1\), \(P_5\), \(\overline{P_5}\), and \(C_5\). For these graphs, the values stored by the algorithm are exactly the values listed above: \(\gamma_R(K_1)=1\), \(\gamma_R^+(K_1)=2\), \(\gamma_R(P_5)=\gamma_R^+(P_5)=4\), \(\gamma_R(\overline{P_5})=\gamma_R^+(\overline{P_5})=3\), and \(\gamma_R(C_5)=\gamma_R^+(C_5)=4\). Hence, the algorithm is correct on all basic nodes.
	
	Assume now that the algorithm stores correct values for all children of a node \(h\), and let \(G_h\) be the graph represented by \(h\). If \(h\) is a disconnected node, then \(G_h\) is the disjoint union of the graphs represented by its children. Since no edge joins two different connected components, every Roman dominating function of \(G_h\) is obtained by taking Roman dominating functions independently on the components. Hence, the value \(\rho(h)\) stored by the algorithm is exactly \(\gamma_R(G_h)\). For \(\rho^+(h)\), at least one component must contain a vertex of value \(2\), while all other components may use ordinary optimum Roman dominating functions. Therefore, the algorithm correctly chooses the component in which the positive condition is imposed, and the stored value \(\rho^+(h)\) is exactly \(\gamma_R^+(G_h)\).
	
	If \(h\) is a complement-disconnected node, then \(G_h\) is the join of the graphs represented by its children. If no vertex receives value \(2\), then every vertex must receive value \(1\), giving weight \(n_h\). If all value-\(2\) vertices lie in a single join-part, then that join-part must realize a positive Roman dominating function, while every vertex outside it may receive value \(0\), because every outside vertex is adjacent to every vertex of that join-part. If value-\(2\) vertices occur in at least two join-parts, then assigning value \(2\) to one vertex in each of two different join-parts dominates the whole graph, and the best cost of this type is \(4\). These cases are exhaustive, so the join formulas used by the algorithm are correct.
	
	If \(h\) is a spider or quasi-spider node, then the node records its underlying partition \(S\cup C\cup R\), and the distinguished child representing \(G_h[R]\), when \(R\neq\emptyset\), has already been processed. The correctness at \(h\) follows from the closed formulas proved above. In a thin spider or thin quasi-spider, the optimum value is \(k+1\), and in a thick spider or thick quasi-spider, the optimum value is \(3\). The proof of these formulas accounts for the possible replacement of one vertex by two twins in a quasi-spider. It also accounts for the set \(R\), because every vertex of \(R\) is complete to the clique side and anticomplete to the independent side. Hence, a value-\(2\) vertex placed on the clique side dominates all of \(R\), and the vertices in \(R\) cannot force a larger value than the stated optimum.
	
	Thus, by induction, every node receives the correct values \(\gamma_R\) and \(\gamma_R^+\). In particular, if \(r\) is the root of the decomposition tree, then the value \(\rho(r)\) returned by Algorithm~\ref{alg:roman-p4-tidy} is exactly \(\gamma_R(G)\).
\end{proof}

\subsection{Running Time Analysis}

\begin{theorem}\label{thm:p4-tidy-roman-running}
	Given the Giakoumakis–Roussel–Thuillier decomposition tree of a \(P_4\)-tidy graph \(G\), Algorithm~\ref{alg:roman-p4-tidy} computes \(\gamma_R(G)\) in time linear in the size of the decomposition tree. Consequently, using the linear-time recognition and decomposition algorithm of Giakoumakis et al.~\cite{giakoumakis1997p_4}, the Roman domination number of a \(P_4\)-tidy graph can be computed in \(O(n+m)\) time.
\end{theorem}

\begin{proof}
	Let \(T_G\) be the decomposition tree. The algorithm processes the nodes of \(T_G\) once, in postorder. At a basic node isomorphic to \(K_1\), \(P_5\), \(\overline{P_5}\), or \(C_5\), the algorithm stores two fixed values. Therefore, each basic node is processed in constant time.
	
	Consider a disconnected node \(h\) with children \(h_1,\ldots,h_t\). The algorithm computes the sum \(A=\sum_{i=1}^{t}\rho(h_i)\) and minimum value \(\delta=\min_{1\leq i\leq t}\{\rho^+(h_i)-\rho(h_i)\}\). Both quantities are obtained by scanning the \(t\) children once. Therefore, the time spent at such a node is \(O(t)\).
	
	Now consider a complement-disconnected node \(h\) with children \(h_1,\ldots,h_t\). The algorithm computes \(B=\min_{1\leq i\leq t}\rho^+(h_i)\), and then evaluates the two formulas \(\rho^+(h)=\min\{4,B\}\) and \(\rho(h)=\min\{n_h,\rho^+(h)\}\). Again, this requires one scan over the \(t\) children and a constant amount of additional work. Hence, the time spent at this node is \(O(t)\).
	
	Finally, consider a spider or quasi-spider node. The decomposition node already specifies whether the node is thin or thick and whether it is a spider or quasi-spider. It also provides the underlying spider partition \(S\cup C\cup R\), together with the distinguished child representing the head when \(R\neq\emptyset\). The algorithm only needs the value \(k=|S|=|C|\) in the underlying spider. If the node is thin, it stores \(\rho(h)=\rho^+(h)=k+1\); if the node is thick, it stores \(\rho(h)=\rho^+(h)=3\). Thus, each spider or quasi-spider node is processed in constant time once the decomposition information is available.
	
	Therefore, the total running time of the bottom-up computation is proportional to the number of basic nodes plus the total number of child pointers over all disconnected and complement-disconnected nodes plus the number of spider and quasi-spider nodes. This is linear in the size of \(T_G\). The decomposition tree of Giakoumakis et al.~\cite{giakoumakis1997p_4} has linear size in the number of vertices. Hence, Algorithm~\ref{alg:roman-p4-tidy} runs in linear time when the decomposition tree is given. If the input graph is given without the decomposition tree, we first construct the Giakoumakis--Roussel--Thuillier \(P_4\)-tidy decomposition tree and then apply the algorithm. The recognition and decomposition procedure of Giakoumakis et al.~\cite{giakoumakis1997p_4} runs in \(O(n+m)\) time for a graph with \(n\) vertices and \(m\) edges. The dynamic-programming part is linear in the size of the decomposition tree, and hence does not dominate this preprocessing time. Therefore, the full algorithm computes \(\gamma_R(G)\) in \(O(n+m)\) time.
	
	In particular, since cographs are \(P_4\)-free and hence \(P_4\)-tidy, this theorem is consistent with the \(O(n+m)\)-time algorithm of Liedloff et al.~\cite{liedloff2008efficient} for cographs, while extending linear-time solvability to the full \(P_4\)-tidy class.
\end{proof}

\section{Conclusion}\label{sec:conclusion}
In this paper, we studied the \textsc{Roman Domination Problem} on three structured graph classes. For circular-convex bipartite graphs, we gave an \(O(n^6)\)-time algorithm by cutting the circular ordering, branching over at most two wrap-around vertices assigned value \(2\), and then applying a boundary-aware dynamic program. This extends the known tractability of Roman domination from convex bipartite graphs to their circular-convex superclass.

For triad-convex bipartite graphs, we exploited the three-arm structure of the underlying subdivision of \(K_{1,3}\). We showed that there exists an optimal Roman dominating function in which at most three central \(Y\)-vertices receive value \(2\). By branching over these vertices and over the value assigned to the center, the three arms can be solved independently, and their solutions can subsequently be combined. This gives an \(O(n^7)\)-time algorithm. Since \textsc{Roman Domination} is NP-complete on broader tree-convex subclasses, such as star-convex and comb-convex bipartite graphs~\cite{padamutham2020algorithmic}, this result identifies triad-convex bipartite graphs as a tractable restricted subclass within a generally hard tree-convex setting.

We also considered \(P_4\)-tidy graphs. Roman domination is already known to be solvable in \(O(n+m)\) time on cographs~\cite{liedloff2008efficient}, which are precisely the \(P_4\)-free graphs. Using the structural decomposition of \(P_4\)-tidy graphs due to Giakoumakis et al.~\cite{giakoumakis1997p_4}, we extended this tractability to the larger class of \(P_4\)-tidy graphs. The resulting direct algorithm stores \(\gamma_R\) and the auxiliary value \(\gamma_R^+\) at each decomposition node and provides explicit rules for disjoint-union, join, spider, quasi-spider, and exceptional nodes. Together with the linear-time recognition and decomposition procedure, this yields an \(O(n+m)\)-time algorithm for \(P_4\)-tidy graphs.

Several directions remain open. It would be interesting to improve the running times of the circular-convex and triad-convex algorithms and to incorporate recognition and representation construction for these two graph classes into the algorithmic framework. Another natural direction is to identify further tractable subclasses of tree-convex bipartite graphs. Finally, the decomposition-based approach for \(P_4\)-tidy graphs may be extendable to broader \(P_4\)-structured graph classes and to other Roman-type domination variants.

\bibliographystyle{plain}
\bibliography{RD_ref}

\end{document}